%% file: main.tex
\documentclass[letterpaper,reqno,11pt]{article}
\usepackage[margin=1.0in]{geometry}
\usepackage{color,latexsym,amsmath,amssymb}
\usepackage{fancyhdr}
\usepackage{amsthm}
\usepackage[linesnumbered,lined,boxed,commentsnumbered,ruled]{algorithm2e}
\usepackage{dsfont}
\usepackage{graphicx}
\usepackage[hidelinks=true]{hyperref}
\usepackage{setspace}
\usepackage{lmodern}
\usepackage{thmtools}
\usepackage{thm-restate}
\usepackage[numbers]{natbib}
\usepackage{subcaption}
\usepackage[inline]{enumitem}
\usepackage{nicematrix}
\usepackage{cleveref}
\usepackage{booktabs}
\usepackage{multirow}
\usepackage{makecell}
\usepackage{dsfont}

\usepackage{tikz}
\usetikzlibrary{bbox}
\usetikzlibrary{fit}
\usetikzlibrary{arrows.meta}

\usepackage{pgfplots}
\pgfplotsset{compat=1.18}
\usepgfplotslibrary{fillbetween}

\usepackage{listings}
\crefname{subsection}{subsection}{subsections}
\Crefname{subsection}{Subsection}{Subsections}

\allowdisplaybreaks

\newcommand{\RR}{\mathbb{R}}

\newcommand{\ZZ}{\mathbb{Z}}
\newcommand{\QQ}{\mathbb{Q}}
\newcommand{\NN}{\mathbb{N}}
\newcommand{\T}{\mathsf{T}}

\newcommand{\NP}{\mathsf{NP}}
\newcommand{\minimize}{\textrm{minimize}}
\newcommand{\maximize}{\textrm{maximize}}
\newcommand{\subjectto}{\textrm{subject to}}

\newcommand{\OPT}{\mathsf{OPT}}
\newcommand{\LPOPT}{\mathsf{LPOPT}}
\newcommand{\DP}{\mathsf{DP}}

\DeclareMathOperator{\Conv}{\mathsf{conv}}
\DeclareMathOperator{\Prob}{Pr}
\DeclareMathOperator{\EE}{\mathbb{E}}
\DeclareMathOperator{\Bin}{\mathsf{Bin}}
\DeclareMathOperator{\Var}{\mathsf{Var}}

\title{A Better-Than-$3$ Approximation Algorithm for Demand Matching via Knapsack Intersection LP and Contention Resolution}
\author{Michel X.\ Goemans \\ \small MIT \\ \small \tt goemans@math.mit.edu \and Yuchong Pan \\ \small MIT \\ \small \tt yuchong@mit.edu}
\date{}
\newtheorem{theorem}{Theorem}[section]
\newtheorem{lemma}[theorem]{Lemma}

\newtheorem{corollary}[theorem]{Corollary}

\newtheorem{proposition}[theorem]{Proposition}

\theoremstyle{definition} 

\begin{document}

\maketitle

\begin{abstract}
    \input{abstract}
\end{abstract}

\input{intro}
\input{knapsack}
\input{simple}
\input{multi}
\input{concluding}
\input{acknowledgements}

\bibliographystyle{abbrvnat}
\bibliography{refs.bib}

\appendix
\crefalias{section}{appendix}

\input{tight}
\input{pareto}
\input{fptas}

\end{document}

%% file: abstract.tex
The demand matching problem generalizes both the knapsack problem and the $b$-matching problem. In this problem, each edge of a graph has a demand and a weight, and each vertex has a capacity. The goal is to find a maximum weight subset of edges such that, at each vertex, the total demand of the incident selected edges does not exceed the vertex capacity. Parekh [IPCO 2011] proved that, if each edge is individually feasible, the natural LP relaxation for demand matching has integrality gap at most $3$, yielding a $3$-approximation algorithm. This bound is tight for the natural LP relaxation, matching the lower bound of Shepherd and Vetta [Math.~Oper.~Res.~2007].

We present a randomized $(3/2 + \sqrt{2} + \varepsilon) \approx (2.914 + \varepsilon)$-approximation algorithm for the demand matching problem for every $\varepsilon > 0$, giving the first approximation ratio strictly better than $3$. For bipartite graphs, we obtain a randomized $(2 + \varepsilon)$-approximation algorithm for every $\varepsilon > 0$. Both algorithms run in time polynomial in $1/\varepsilon$ and the input length. Our algorithms use a strengthened LP relaxation based on intersecting the integral knapsack polytopes associated with the vertices, together with a multiple-choice generalization. As a key ingredient, we prove the existence of a $(q, 1/(1+q))$-balanced contention resolution scheme for the integral knapsack polytope for every $q \in [0, 1]$, which may be of independent interest. The balance guarantee $1/(1+q)$ is tight in the worst case over all knapsack instances.

%% file: intro.tex
\section{Introduction} \label{sec:intro}

We consider the \textsc{Demand Matching} problem introduced by \citet{shepherd2007demand}. Let $G = (V, E)$ be a (multi)graph, where multiple edges between the same pair of vertices are allowed. Let $b : V \to \ZZ_+$ denote \emph{vertex capacities}. Let $d : E \to \NN$ and $w : E \to \ZZ_+$ denote \emph{edge demands} and \emph{edge weights}, respectively.\footnote{We use $\NN$ and $\ZZ_+$ to denote the set of positive integers and the set of nonnegative integers, respectively.} We say that a subset $M \subseteq E$ is a \emph{demand matching} if
\begin{align*}
    d(M \cap \delta(v)) \leq b(v) && \forall v \in V,
\end{align*}
where $\delta(v)$ denotes the set of edges of $G$ incident to $v$.\footnote{Given a function $\varphi : \Sigma \to \RR$ on a ground set $\Sigma$ and a subset $S \subseteq \Sigma$, we define $\varphi(S) := \sum_{e \in S} \varphi(e)$.} The objective of the \textsc{Demand Matching} problem is to find a demand matching $M$ such that $w(M)$ is maximized. Without loss of generality, we assume that $G$ does not contain self-loops.

The \textsc{Demand Matching} problem generalizes two classic combinatorial optimization problems. In the unit-demand case, where $d(e) = 1$ for all $e \in E$, it reduces to the \textsc{Maximum Weight $b$-Matching} problem, for which several polynomial-time algorithms are known \cite{marsh1979matching,padberg1982odd,gabow1983efficient,anstee1987polynomial,gerards1995matching}. On the other hand, when $G$ is a star and every leaf capacity is at least the demand of its incident edge (or, equivalently, when $G$ is a two-vertex multigraph), the problem reduces to the \textsc{$0$-$1$ Knapsack} problem, which is $\NP$-hard. Thus, the \textsc{Demand Matching} problem is $\NP$-hard.

The \emph{natural LP relaxation} of the \textsc{Demand Matching} problem is obtained by replacing each $\{ 0, 1 \}$-valued variable with a fractional variable in $[0, 1]$:
\begin{alignat}{4}
    \maximize \qquad && \sum_{e \in E} x_e w(e) & \tag{Natural-LP} \label{eq:lp} \\
    \subjectto \qquad && \sum_{e \in \delta(v)} x_e d(e) & \leq b(v) && \qquad \forall v \in V, \notag \\
    && x_e &\in [0, 1] && \qquad \forall e \in E. \notag
\end{alignat}
This relaxation has been extensively studied under the \emph{no-bottleneck assumption} which states that
\begin{align*}
    d(e) \leq \min \{ b(u), b(v) \} && \forall e = uv \in E.
\end{align*}
This assumption can be made, without loss of generality, without changing the integral optimum, as an edge $e = uv \in E$ violating $d(e) \leq \min\{ b(u), b(v) \}$ cannot belong to any feasible demand matching and can therefore be deleted. Deleting these edges may, however, change the integrality gap of an LP relaxation. Indeed, without this assumption, the integrality gap of the natural LP relaxation is unbounded. In \Cref{tab:dm}, we summarize known lower and upper bounds on the integrality gap of the natural LP relaxation under the no-bottleneck assumption.\footnote{In a preliminary version, \citet{shepherd2007demand} proved an upper bound of $3.3125$ for general multigraphs, and an upper bound of $2.8125$ for bipartite multigraphs.} On the hardness side, \citet{shepherd2007demand} showed that the \textsc{Demand Matching} problem is $\mathsf{MAXSNP}$-hard even in the cardinality case (i.e., $w(e) = 1$ for all $e \in E$).

\begin{table}[ht]
    \centering
    \begin{tabular}{llll}
        \hline
        \textbf{Graph class} & \textbf{Lower bound} & \textbf{Upper bound} & \textbf{Reference} \\
        \hline
        General multigraphs
            & $3$
            & $3.264$
            & \cite{shepherd2007demand} \\
        General multigraphs
            &
            & $3$
            & \cite{parekh2011iterative} \\
        Bipartite multigraphs
            & $2.5$
            & $2.764$
            & \cite{shepherd2007demand} \\
        Bipartite simple graphs
            & $2.699$
            & $2.709$
            & \cite{singh2012nearly} \\
        \hline
    \end{tabular}
    \caption{Known lower and upper bounds on the integrality gap of the natural LP relaxation of the \textsc{Demand Matching} problem under the no-bottleneck assumption.}
    \label{tab:dm}
\end{table}

The main contribution of this paper is the following theorem.

\begin{theorem} \label{thm:main}
    For every $\varepsilon > 0$,
    \begin{itemize}[itemsep=0pt]
        \item there is a fully polynomial-time randomized $(3/2 + \sqrt{2} + \varepsilon) \approx (2.914 + \varepsilon)$-approximation algorithm for the \textsc{Demand Matching} problem;\footnote{We say that an algorithm is \emph{fully polynomial-time} with respect to a parameter $\varepsilon > 0$ if its running time is polynomial in the input length and $1/\varepsilon$.}\footnote{For $\alpha \geq 1$, we say that a randomized algorithm is a \emph{randomized $\alpha$-approximation algorithm} for a maximization problem with optimum $\OPT$ if, given an input, the expected objective value of its output is at least $\OPT/\alpha$.}
        \item there is a fully polynomial-time randomized $(2 + \varepsilon)$-approximation algorithm for the special case of the \textsc{Demand Matching} problem where the input graph is bipartite.
    \end{itemize}
\end{theorem}

To the best of our knowledge, this gives the first approximation ratio strictly better than $3$ for the \textsc{Demand Matching} problem. Our approximation guarantees in \Cref{thm:main} are relative to a strengthened LP relaxation rather than the natural LP relaxation.

In the following subsections, we provide a brief overview of our algorithms and key ingredients, and introduce the strengthened LP relaxations used by our algorithms.

\subsection{Overview of the Algorithms}

The core idea for proving \Cref{thm:main} is best illustrated in the case of simple graphs, where there exists at most one edge between any pair of vertices. The main advantage of simple graphs is that no two edges incident to the same vertex are associated with the same neighbor, which gives the independence needed in our rounding procedure. In a multigraph, parallel edges share both endpoints and may therefore be correlated in the rounding procedure, so this independence property no longer holds. To handle these correlations, we develop additional machinery in \cref{sec:multi}, including bundling parallel edges, approximating bundle families, and extending the analysis to multiple-choice sampling.

\paragraph{Knapsack Intersection LP.} The motivation is to strengthen the natural LP relaxation of the \textsc{Demand Matching} problem as follows. In \eqref{eq:lp}, the constraint for each $v \in V$,
$$ \sum_{e \in \delta(v)} x_e d(e) \leq b(v), $$
together with $x \in [0, 1]^E$, defines the \emph{fractional knapsack polytope} on $\delta(v)$ with item sizes $d$ and capacity $b(v)$. We replace this constraint by the requirement that $x$ belongs to the corresponding \emph{integral knapsack polytope}, i.e.,
$$ x \in K(v) := \Conv \left\{ \chi^S : S \subseteq E, d(S \cap \delta(v)) \leq b(v) \right\}. $$
This gives the following LP relaxation of the \textsc{Demand Matching} problem, which we call the \emph{knapsack intersection LP}:
\begin{alignat}{4}
    \maximize \qquad && \sum_{e \in E} x_e w(e) & \tag{Knapsack-LP} \label{eq:knapsack-lp} \\
    \subjectto \qquad && x \in K(v) && \qquad \forall v \in V. \notag
\end{alignat}

The knapsack intersection LP is precisely the first level of the knapsack intersection hierarchy introduced by \citet*{jozefiak2023knapsack} for packing integer programs. More generally, at level $t$, their hierarchy strengthens the natural LP relaxation of a packing integer program by imposing the integer hull of every subsystem consisting of $t$ packing constraints. In other words, their hierarchy captures the maximum possible strength of cuts obtained by considering $t$ packing constraints at a time, and interpolates between the natural LP relaxation and the full integer hull.

Exact linear optimization over \eqref{eq:knapsack-lp} is $\NP$-hard. Indeed, when $G$ is a star and each leaf has capacity at least the demand of its incident edge, the feasibility region of \eqref{eq:knapsack-lp} coincides with the integral knapsack polytope at the center. Optimizing over this polytope is equivalent to solving the \textsc{$0$-$1$ Knapsack} problem, which is $\NP$-hard. Nevertheless, building on a technique of \citet{pritchard2010lp}, \citet{jozefiak2023knapsack} showed that linear optimization over \eqref{eq:knapsack-lp} (and more generally, over every fixed level of their hierarchy) has a polynomial-time approximation scheme (PTAS). In fact, as we show in \Cref{lem:fptas}, a multiple-choice generalization of \eqref{eq:knapsack-lp} admits a fully polynomial-time approximation scheme (FPTAS). This formulation specializes to \eqref{eq:knapsack-lp} on simple graphs.

\paragraph{Contention Resolution Scheme for the Integral Knapsack Polytope.} A key ingredient in the proof of \Cref{thm:main} for simple graphs is the existence of a $(q, 1/(1+q))$-balanced contention resolution scheme for the integral knapsack polytope, for every $q \in [0, 1]$. We believe that this result may be of independent interest.

We first recall the definition of a contention resolution scheme introduced by \citet*{chekuri2014submodular}. Let $N$ be a finite ground set. Let $\mathcal I \subseteq 2^N$ be a nonempty downward-closed family of \emph{feasible} subsets of $N$.\footnote{We say that a family $\mathcal I \subseteq 2^N$ on a finite ground set $N$ is \emph{downward-closed} if, for all $A \subseteq B \subseteq N$, $B \in \mathcal I$ implies $A \in \mathcal I$.} Let $P_{\mathcal I}$ be a \emph{convex relaxation} of $\mathcal I$, i.e., $P_{\mathcal I}$ is convex and
$$ \Conv\left\{ \chi^S : S \in \mathcal I \right\} \subseteq P_{\mathcal I} \subseteq [0, 1]^N. $$
Given $x \in P_{\mathcal I}$, a natural strategy for rounding $x$ into an integral solution is to independently select each element $e \in N$ with probability $x_e$. The resulting set, however, need not be feasible. A contention resolution scheme enforces feasibility by removing some of the selected elements while retaining each element with sufficiently large probability.

Formally, given $y \in [0, 1]^N$, let $R(y) \subseteq N$ denote the random subset obtained by independently including each element $e \in N$ with probability $y_e$. Given $q, p \in [0, 1]$, a \emph{$(q, p)$-balanced contention resolution (CR) scheme} for $P_{\mathcal I}$ is a randomized scheme that, given any $x \in P_{\mathcal I}$, outputs a feasible subset $I \subseteq R(qx)$ such that for all $e \in N$ with $qx_e > 0$,
$$ \Pr[e \in I \mid e \in R(qx)] \geq p. $$

\citet[Theorem 4.3]{chekuri2014submodular} proved the following equivalent characterization of a $(q, p)$-balanced contention resolution scheme: there exists a $(q, p)$-balanced contention resolution scheme for $P_{\mathcal I}$ if and only if, for all $x \in P_{\mathcal I}$ and $w \in \RR_+^N$,
$$ \EE[\max\{ w(S) : S \subseteq R(qx), S \in \mathcal I \}] \geq pq w^{\T} x. $$
In other words, after sampling the random subset $R(qx)$, one can extract a feasible subset whose expected weight is at least a $p$-fraction of the expected weight of $R(qx)$.

Given item sizes $d_1, \ldots, d_n \in \NN$ and a knapsack capacity $B \in \ZZ_+$, the \emph{knapsack independence system} is defined to be\footnote{Here, we define $d(S) := \sum_{e \in S} d_e$ for all $S \subseteq [n]$.}
$$ \mathcal K := \{ S \subseteq [n] : d(S) \leq B \}, $$
and the associated \emph{fractional} and \emph{integral knapsack polytopes} are defined to be
\begin{gather*}
    P_{\mathcal K}^{\mathsf{frac}} := \left\{ x \in [0, 1]^n : \sum_{e \in [n]} d_e x_e \leq B \right\}, \qquad P_{\mathcal K}^{\mathsf{int}} := \Conv\left\{ \chi^S : S \in \mathcal K \right\},
\end{gather*}
respectively. For the fractional knapsack polytope after discarding oversized items, \citet{chekuri2014submodular} gave a $(q, 1 - 2q)$-balanced contention resolution scheme for all $q \in [0, 1/2]$. In this paper, we instead work with the integral knapsack polytope, which is the exact convex hull of feasible knapsack solutions. Since it captures all constraints valid for integral knapsack solutions, one may hope to preserve a larger fraction of the sampled weight during contention resolution. We show that this is indeed possible, proving the following improved guarantee.

\begin{theorem} \label{thm:knapsack}
    For all $q \in [0, 1]$, there exists a $(q, 1/(1+q))$-balanced contention resolution scheme for the integral knapsack polytope.
\end{theorem}

The balance guarantee $1/(1+q)$ is tight in the worst case: for every $q \in (0, 1]$ and every $p > 1/(1+q)$, there exists a knapsack instance whose integral knapsack polytope admits no $(q, p)$-balanced contention resolution scheme. We defer the proof to \cref{apx:tight}.

\begin{proposition} \label{prop:tight}
    Let $q \in (0, 1]$. For every $p > 1/(1+q)$, there exists a knapsack instance whose integral knapsack polytope admits no $(q, p)$-balanced contention resolution scheme. Equivalently, for every $\varepsilon > 0$, there exist $n \in \NN$, $d_1, \ldots, d_n \in \NN$, $B \in \ZZ_+$, $x \in P_{\mathcal K}^{\mathsf{int}}$ and $w \in \RR_+^n$ such that
    $$ \EE\left[\max\left\{ w(S) : S \subseteq R(qx), d(S) \leq B \right\}\right] < q\left(\frac{1}{1 + q}+\varepsilon\right) w^{\T} x. $$
\end{proposition}

\paragraph{Rounding Algorithms for Simple Graphs.} We now give a brief overview of the rounding algorithms for simple graphs. For simplicity, the discussion of polynomial-time implementation is deferred to \cref{subsec:polytime}.

Let $x \in \RR^E$ be a feasible solution of \eqref{eq:knapsack-lp}. First, let us consider the case where the graph is simple and bipartite, with bipartition $U \sqcup W$. We say that vertices in $U$ are \emph{sources} and those in $W$ are \emph{sinks}.

For each source $u \in U$, since $x \in K(u)$, we can express $x|_{\delta(u)}$ as a convex combination of (the characteristic vectors of) subsets $S \subseteq \delta(u)$ such that $d(S) \leq b(u)$; in other words, there exists a distribution $\mathcal D_u$ over such subsets satisfying
\begin{align*}
    \Pr_{S \sim \mathcal D_u}[e \in S] = x_e && \forall e \in \delta(u).
\end{align*}
Independently for each $u \in U$, we sample a subset $S_u \sim \mathcal D_u$. One may think of each source $u \in U$ \emph{proposing} to its neighboring sinks a subset $S_u \subseteq \delta(u)$ that satisfies the capacity constraint at $u$.

Now, for each sink $v \in W$, we aggregate all proposed edges from its neighboring sources, and let $R_v \subseteq \delta(v)$ denote this subset of proposed edges, i.e.,
$$ R_v := \bigcup_{u \in U} \left(S_u \cap \delta(v)\right). $$
Since the graph is simple, distinct edges in $R_v$ are proposed by distinct neighboring sources. Since the sets $S_u$ are sampled independently across sources, it follows that $R_v$ has exactly the distribution $R(x|_{\delta(v)})$. In other words, each edge $e \in \delta(v)$ is included in $R_v$ independently with probability $x_e$.

Hence, we can apply \Cref{thm:knapsack} to obtain a maximum weight subset $T_v \subseteq R_v$ that satisfies the capacity constraint at $v$ and that has expected weight at least $1/2 \cdot \sum_{e \in \delta(v)} w(e) x_e$. Taking the union of $T_v$ over all sinks $v \in W$ yields a random demand matching with expected weight at least
$$ \frac{1}{2} \sum_{e \in E} w(e) x_e. $$
Therefore, the integrality gap of \eqref{eq:knapsack-lp} is at most $2$ for simple bipartite graphs.

General simple graphs do not have an intrinsic bipartition of vertices into sources and sinks. We therefore independently label each vertex as a source with probability $q$ and as a sink otherwise, for some parameter $q \in [0, 1]$ to be determined later. We then apply the same procedure as in the simple bipartite case. Each source $u$ proposes a random subset $S_u \subseteq \delta(u)$ with marginals $x|_{\delta(u)}$ that satisfies the capacity constraint at $u$. Conditioned on a vertex $v$ being a sink, it receives a random subset $R_v \subseteq \delta(v)$ in which each edge $e \in \delta(v)$ is included independently with probability $qx_e$. By \Cref{thm:knapsack}, we can select a maximum weight subset $T_v \subseteq R_v$ that satisfies the capacity constraint at $v$ and whose expected weight is at least a $q/(1+q)$-fraction of the fractional weight of the edges incident to $v$. Taking the union of $T_v$ over all sinks $v$ yields a demand matching with expected weight at least
$$ \frac{2q(1 - q)}{1 + q} \sum_{e \in E} w(e) x_e. $$
Hence, the integrality gap of \eqref{eq:knapsack-lp} on simple graphs is at most $(1 + q)/(2q(1-q))$. This quantity is minimized at $q = \sqrt{2} - 1$, giving an upper bound of
$$ \frac{3}{2} + \sqrt{2} \approx 2.914 $$
on the integrality gap of \eqref{eq:knapsack-lp}.

\subsection{Organization of This Paper}

This paper is organized as follows. In \cref{sec:knapsack}, we prove \Cref{thm:knapsack}. In \cref{sec:simple}, we prove \Cref{thm:main} for simple graphs. In \cref{sec:multi}, we introduce the additional machinery and extend the results to multigraphs. In \cref{sec:concluding}, we conclude with several remarks and open questions.

%% file: knapsack.tex
\section{CR Scheme for the Integral Knapsack Polytope} \label{sec:knapsack}

In this section, we prove \Cref{thm:knapsack}. Let $d_1, \ldots, d_n \in \NN$ and $B \in \ZZ_+$. Let $w \in \RR_+^n$. Recall that
\begin{gather*}
    \mathcal K := \left\{ S \subseteq [n] : d(S) \leq B \right\}, \qquad P_{\mathcal K}^{\mathsf{int}} := \Conv\left\{ \chi^S : S \in \mathcal K \right\}.
\end{gather*}
Let $x \in P_{\mathcal K}^{\mathsf{int}}$. Let $q \in [0, 1]$. Let $R := R(qx) \subseteq [n]$ be the random set produced by the sampling process, and let $r(R, B)$ denote the weight of a maximum weight feasible knapsack subset of $R$, i.e.,
$$ r(R, B) := \max \{ w(T) : T \subseteq R, d(T) \leq B \}. $$
The goal is to prove that
$$ \EE[r(R, B)] \geq \frac{q}{1 + q} \cdot w^{\T} x. $$

We begin with the intuition of the proof. Since $x \in P_{\mathcal K}^{\mathsf{int}}$, it can be interpreted as the marginals of a distribution over feasible subsets $S \subseteq [n]$. Our goal is then to compare the weight of a feasible subset $S$ from this distribution with $r(R, B)$. To this end, we insert the elements of $S$ into $R$ one at a time. After all insertions, the resulting set contains the feasible subset $S$, so its optimal knapsack value is at least $w(S)$. Hence, the total increase in the optimal knapsack value in this process is at least the gap between $w(S)$ and $r(R, B)$. We measure the increase from each element by its \emph{deletion sensitivity}, and prove two properties of this quantity:
\begin{itemize}[itemsep=0pt]
    \item {\bf \em Monotonicity:} deletion sensitivities do not decrease when other items are removed.
    \item {\bf \em Bounded total sensitivity:} the total deletion sensitivity over a set is at most the optimal knapsack value of that set.
\end{itemize}
These two properties allow us to bound the total increase of inserting $S$ into $R$. Finally, averaging over $S$ and $R$ yields the desired bound.

Formally, we introduce the following definitions. For $S \subseteq [n]$ and $k \in \ZZ_+$, we define
$$ r(S, k) := \max \{ w(T) : T \subseteq S, d(T) \leq k \}. $$
For $S \subseteq [n]$, $e \in [n]$, and $\ell \in \ZZ_+$, we define the \emph{deletion sensitivity} of $e$ with respect to $S$ and capacity $\ell$ to be
$$ \Delta_e(S, \ell) := \max_{k \in \{ 0, \ldots, \ell \}} (r(S, k) - r(S \setminus \{ e \}, k)). $$
In other words, $\Delta_e(S, \ell)$ measures the maximum decrease in the optimal knapsack value, over all capacities up to $\ell$, caused by deleting $e$ from $S$. Taking this maximum ensures the deletion monotonicity property proved below, while the deletion loss at a fixed capacity can decrease when other items are removed. Some properties of $\Delta_e(S, \ell)$ are immediate: $\Delta_e(S, \ell)$ is nonnegative and nodecreasing in $\ell$, and $\Delta_e(S, \ell) = 0$ for all $e \in [n] \setminus S$ and $\ell \in \ZZ_+$.

Now, we state and prove the two properties mentioned above.

\begin{lemma}[Deletion monotonicity] \label{lem:mono}
    Let $S \subseteq [n]$. Let $e, f \in S$ be distinct. Let $\ell \in \ZZ_+$. Then
    $$ \Delta_e(S, \ell) \leq \Delta_e(S \setminus \{ f \}, \ell). $$
\end{lemma}

\begin{proof}
    Let $T_1 := S \setminus \{ f \}$ and $T_2 := S \setminus \{ e, f \}$. By definition, for all $k \in \{ 0, \ldots, \ell \}$,
    \begin{equation} \label{eq:Delta}
        r(T_1, k) \leq r(T_2, k) + \Delta_e(T_1, \ell).
    \end{equation}
    For all $k \in \{ 0, \ldots, \ell \}$,
    $$ r(S, k) = \max\{ r(T_1, k), r(T_1, k - d_f) + w_f \}, $$
    and
    $$ r(S \setminus \{ e \}, k) = \max\{ r(T_2, k), r(T_2, k - d_f) + w_f \}, $$
    where the second term in each maximum is omitted if $k < d_f$. By \eqref{eq:Delta}, for all $k \in \{ 0, \ldots, \ell \}$,
    $$ r(S, k) \leq r(S \setminus \{ e \}, k) + \Delta_e(T_1, \ell). $$
    Taking the maximum over $k \in \{ 0, \ldots, \ell \}$ completes the proof.
\end{proof}

\begin{corollary} \label{cor:mono}
    Let $S, T \subseteq [n]$ be such that $T \subseteq S$. Let $e \in T$. Let $\ell \in \ZZ_+$. Then
    $$ \Delta_e(S, \ell) \leq \Delta_e(T, \ell). $$
\end{corollary}

\begin{lemma}[Bounded total deletion sensitivity] \label{lem:del}
    Let $S \subseteq [n]$. Let $\ell \in \ZZ_+$. Then
    $$ \sum_{e \in S} \Delta_e(S, \ell) \leq r(S, \ell). $$
\end{lemma}

\begin{proof}
    We proceed by induction on $|S|$. The base case $S = \emptyset$ is trivial. For the induction step, suppose that $S \neq \emptyset$ and that $\sum_{e \in T} \Delta_e(T, k) \leq r(T, k)$ for all $T \subseteq [n]$ with $|T| < |S|$ and for all $k \in \ZZ_+$. For $e \in S$, let $k_e \in \{ 0, \ldots, \ell \}$ be such that $\Delta_e(S, \ell) = r(S, k_e) - r(S \setminus \{ e \}, k_e)$. Let $f \in S$ be such that $k_f = k := \max_{e \in S} k_e$. Then $\Delta_e(S, \ell) = \Delta_e(S, k)$ for all $e \in S$, and $\Delta_f(S, k) = r(S, k) - r(S \setminus \{ f \}, k)$. Hence,
    \begin{align*}
        \sum_{e \in S} \Delta_e(S, \ell) &= \Delta_f(S, k) + \sum_{e \in S \setminus \{ f \}} \Delta_e(S, k) \\
        &\leq r(S, k) - r(S \setminus \{ f \}, k) + \sum_{e \in S \setminus \{ f \}} \Delta_e(S \setminus \{ f \}, k) \\
        &\leq r(S, k) - r(S \setminus \{ f \}, k) + r(S \setminus \{ f \}, k) \\
        &= r(S, k) \\
        &\leq r(S, \ell),
    \end{align*}
    where the first inequality follows from \Cref{lem:mono}, and where the second inequality follows from the inductive hypothesis. This completes the proof.
\end{proof}

Next, we formalize the insertion argument above in the following lemma, which follows from deletion monotonicity. For each $T \subseteq [n]$ and $e \in [n]$, we define
\begin{equation} \label{eq:insert}
    T \vee e := T \cup \{ e \}.
\end{equation}

\begin{lemma}[Insertion lemma] \label{lem:insertion}
    Let $S, T \subseteq [n]$ be such that $d(S) \leq B$. Then
    $$ w(S) \leq r(T, B) + \sum_{e \in S} \Delta_e(T \vee e, B). $$
\end{lemma}

\begin{proof}
    Suppose that $S = \{ e_1, \ldots, e_m \}$. For each $j \in \{ 0, \ldots, m \}$, let $S_j := T \cup \{ e_1, \ldots, e_j \}$. For each $j \in [m]$,
    $$ r(S_j, B) - r(S_{j - 1}, B) \leq \Delta_{e_j}(S_j, B) \leq \Delta_{e_j}(T \vee e_j, B), $$
    where the second inequality follows from \Cref{cor:mono} and from the fact that $T \vee e_j \subseteq S_j$. Since $S \subseteq T \cup S$ and since $d(S) \leq B$,
    \begin{align*}
        w(S) &\leq r(T \cup S, B) \\
        &= r(S_0, B) + r(S_m, B) - r(S_0, B) \\
        &= r(T, B) + \sum_{j = 1}^m \left(r(S_j, B) - r(S_{j - 1}, B)\right) \\
        &\leq r(T, B) + \sum_{j = 1}^m \Delta_{e_j}(T \vee e_j, B) \\
        &= r(T, B) + \sum_{e \in S} \Delta_e(T \vee e, B).
    \end{align*}
    This completes the proof.
\end{proof}

Now, we are ready to prove \Cref{thm:knapsack} using \Cref{lem:del,lem:insertion}.

\begin{proof}[Proof of \Cref{thm:knapsack}]
    Since $x \in P_{\mathcal K}^{\mathsf{int}}$, there exists a distribution $\mathcal D$ over subsets $S \subseteq [n]$ with $d(S) \leq B$ such that
    $$ \Pr_{S \sim \mathcal D}[e \in S] = x_e. $$
    For each $S$ in the support of $\mathcal D$, \Cref{lem:insertion} implies that, after taking expectation over $R$,
    \begin{equation} \label{eq:insertion}
        w(S) \leq \EE_R[r(R, B)] + \sum_{e \in S} \EE_R\left[\Delta_e(R \vee e, B)\right].
    \end{equation}
    Averaging \eqref{eq:insertion} over $S \sim \mathcal D$ yields that
    \begin{align*}
        w^{\T} x &= \EE_{S \sim \mathcal D}[w(S)] \\
        &\leq \EE_R[r(R, B)] + \EE_{S \sim \mathcal D}\left[\sum_{e \in S} \EE_R\left[\Delta_e(R \vee e, B)\right]\right] \\
        &= \EE_R[r(R, B)] + \sum_{e \in [n]} \Pr_{S \sim \mathcal D}[e \in S] \cdot \EE_R\left[\Delta_e(R \vee e, B)\right] \\
        &= \EE_R[r(R, B)] + \sum_{e \in [n]} x_e \EE_R\left[\Delta_e(R \vee e, B)\right].
    \end{align*}
    For each $e \in [n]$, we observe that $\Delta_e(R \vee e, B)$ is independent of whether $e \in R$. Moreover, we have $R \vee e = R$ for each $e \in R$. Hence, for each $e \in [n]$,
    $$ qx_e \EE_R\left[\Delta_e(R \vee e, B)\right] = \Prob_R[e \in R] \cdot \EE_R\left[\Delta_e(R \vee e, B)\right] = \EE_R\left[\mathds 1[e \in R] \cdot \Delta_e(R, B)\right]. $$
    Therefore,
    \begin{align*}
        qw^{\T} x &\leq q\EE_R[r(R, B)] + \sum_{e \in [n]} qx_e \EE_R\left[\Delta_e(R \vee e, B)\right] \\
        &= q\EE_R[r(R, B)] + \sum_{e \in [n]} \EE_R\left[\mathds 1[e \in R] \cdot \Delta_e(R, B)\right] \\
        &= q\EE_R[r(R, B)] + \EE_R\left[\sum_{e \in R} \Delta_e(R, B)\right] \\
        &\leq q\EE_R[r(R, B)] + \EE_R[r(R, B)] \\
        &= (1 + q) \EE_R[r(R, B)],
    \end{align*}
    where the second inequality follows from \Cref{lem:del}. Rearranging completes the proof.
\end{proof}

%% file: simple.tex
\section{Rounding Algorithms for Simple Graphs} \label{sec:simple}

In this section, we prove \Cref{thm:main} for simple graphs. We first consider general simple graphs, where we create a source-sink structure by independently labeling each vertex as a source with probability $q$ and as a sink otherwise. This yields a $(3/2 + \sqrt{2} + \varepsilon)$-approximation, for every $\varepsilon > 0$, by optimizing over $q$. We then consider bipartite simple graphs, where the given bipartition provides a natural, deterministic source-sink structure that eliminates the loss from random labeling, yielding a $(2 + \varepsilon)$-approximation for every $\varepsilon > 0$.

We formulate the rounding algorithms below with a feasible solution $x$ of \eqref{eq:knapsack-lp} given as input. For the integrality gap analysis, we take $x$ to be an optimal solution of \eqref{eq:knapsack-lp}. For clarity, we first present and analyze ``idealized'' versions of these rounding algorithms, which use exact convex decompositions and knapsack optimizations. These idealized versions establish upper bounds of $3/2 + \sqrt{2}$ and $2$ on the integrality gap of \eqref{eq:knapsack-lp} for general and bipartite simple graphs, respectively. In \cref{subsec:polytime}, we obtain fully polynomial-time implementations by first computing a near-optimal feasible solution of \eqref{eq:knapsack-lp} and then implementing the rounding steps approximately, with an arbitrarily small additional loss in the approximation guarantees.

\subsection{General Simple Graphs}

Informally, given a feasible solution $x$ of \eqref{eq:knapsack-lp} and a parameter $q \in [0, 1]$ to be determined later, our algorithm independently labels each vertex $v \in V$ as a \emph{source} with probability $q$ and as a \emph{sink} otherwise. For each source $u$, since $x \in K(u)$, the algorithm samples a subset $S_u \subseteq \delta(u)$ such that $d(S_u) \leq b(u)$ and $\Pr[e \in S_u] = x_e$ for all $e \in \delta(u)$. For each sink $v$, let
$$ R_v := \left\{ e = uv \in \delta(v) : u \text{ is a source}, e \in S_u \right\} $$
be the set of incident edges sampled by neighboring sources. Since the graph is simple, distinct edges in $\delta(v)$ are associated with distinct neighboring sources. Since the sets $S_u$ are sampled independently across sources, conditioned on $v$ being a sink, each edge $e \in \delta(v)$ is included in $R_v$ independently with probability $qx_e$. The algorithm then takes a maximum weight subset $T_v \subseteq R_v$ such that $d(T_v) \leq b(v)$. Finally, the algorithm returns the union of the sets $T_v$ over all sinks $v$. We formally state the algorithm in \Cref{alg:simple}.

\begin{algorithm}[ht]
    \caption{A rounding algorithm for \textsc{Demand Matching} on simple graphs.}
    \label{alg:simple}
    \KwIn{a simple graph $G = (V, E)$, $b : V \to \ZZ_+$, $d : E \to \NN$, $w : E \to \ZZ_+$, $q \in [0, 1]$, and a feasible solution $x$ of \eqref{eq:knapsack-lp}.}
    \KwOut{a subset $M \subseteq E$.}
    \ForEach{$v \in V$ independently}{
        Let $\mathcal D_v$ be a distribution of $S \subseteq \delta(v)$ with $d(S) \leq b(v)$ and $\Pr[e \in S] = x_e$ for $e \in \delta(v)$. \\
        Label $v$ as a \emph{source} with probability $q$ and as a \emph{sink} otherwise.
    }
    \ForEach{source $u \in V$ independently}{
        Sample $S_u \sim \mathcal D_u$ independent of the vertex labels.
    }
    \ForEach{sink $v \in V$}{
        Let $R_v := \{ e = uv \in \delta(v) : u \text{ is a source}, e \in S_u \}$. \\
        Find a maximum weight $T_v \subseteq R_v$ such that $d(T_v) \leq b(v)$.
    }
    \Return{$M := \bigcup \{ T_v : v \in V \text{ sink} \}$}
\end{algorithm}

We analyze \Cref{alg:simple} on an input $(G = (V, E), b, d, w, x)$ where $G$ is simple.

\begin{lemma} \label{lem:simple-feasible}
    Let $M \subseteq E$ be the output of \Cref{alg:simple}. Then $d(M \cap \delta(v)) \leq b(v)$ for all $v \in V$.
\end{lemma}

\begin{proof}
    For each source $u \in V$, since $M \cap \delta(u) \subseteq S_u$, we have $d(M \cap \delta(u)) \leq d(S_u) \leq b(u)$. For each sink $v \in V$, we have $M \cap \delta(v) = T_v$, which is guaranteed to satisfy $d(T_v) \leq b(v)$ by the algorithm. This completes the proof.
\end{proof}

\begin{lemma} \label{lem:simple-weight}
    Let $M \subseteq E$ be the output of \Cref{alg:simple}. Then
    $$ \EE[w(M)] \geq w^{\T} x \cdot \frac{2q(1-q)}{1+q}. $$
    Choosing $q = \sqrt{2} - 1$ gives that $\EE[w(M)] \geq (6 - 4\sqrt{2}) \cdot w^{\T} x$.
\end{lemma}

\begin{proof}
    Assume without loss of generality that $0 < q < 1$. Let $v \in V$. Since $x$ is a feasible solution of \eqref{eq:knapsack-lp}, we have $x \in K(v)$. Each neighbor $u$ is independently a source with probability $q$. Conditioned on $v$ being a sink and $u$ being a source, each edge $e = uv \in \delta(v)$ belongs to $S_u$ with probability $x_e$. Since $G$ is simple, distinct edges incident to $v$ correspond to distinct neighbors. Since the corresponding sets $S_u$ are sampled independently across sources, conditioned on $v$ being a sink, each edge $e \in \delta(v)$ is independently included in $R_v$ with probability $qx_e$. In other words, conditioned on $v$ being a sink, the random set $R_v$ is distributed as $R(qx|_{\delta(v)})$. Since $T_v$ is a maximum weight subset of $R_v$ with $d(T_v) \leq b(v)$, \Cref{thm:knapsack} and the contention-resolution characterization in \cref{sec:intro} imply that
    $$ \EE[w(T_v) \mid v \text{ is a sink}] \geq \frac{q}{1 + q} \sum_{e \in \delta(v)} w(e) x_e. $$
    
    Therefore,
    \begin{align*}
        \EE[w(M)] &= \EE\left[\sum_{v \in V \text{ sink}} w(T_v)\right] \\
        &= \sum_{v \in V} \Pr[v \text{ is a sink}] \cdot \EE[w(T_v) \mid v \text{ is a sink}] \\
        &\geq \sum_{v \in V} (1 - q) \cdot \frac{q}{1 + q} \sum_{e \in \delta(v)} w(e) x_e \\
        &= \frac{2q(1 - q)}{1 + q} \sum_{e \in E} w(e) x_e.
    \end{align*}
    This completes the proof.
\end{proof}

\Cref{lem:simple-feasible,lem:simple-weight} imply that the integrality gap of \eqref{eq:knapsack-lp} on simple graphs is at most $1/(6 - 4\sqrt{2}) = 3/2 + \sqrt{2} \approx 2.914$.

\subsection{Bipartite Simple Graphs}

Next, we consider the case where $G = (U \sqcup W, E)$ is simple and bipartite. Here, the bipartition provides a natural source-sink structure, where we \emph{deterministically} label each vertex in $U$ as a \emph{source} and each vertex in $W$ as a \emph{sink}, thereby avoiding the loss from random labeling. The remainder of the rounding algorithm is identical to \Cref{alg:simple}. We formally state the algorithm in \Cref{alg:bip-simple}.

\begin{algorithm}[ht]
    \caption{A rounding algorithm for \textsc{Demand Matching} on bipartite simple graphs.}
    \label{alg:bip-simple}
    \KwIn{a bipartite simple graph $G = (U \sqcup W, E)$, $b : U \sqcup W \to \ZZ_+$, $d : E \to \NN$, $w : E \to \ZZ_+$, and a feasible solution $x$ of \eqref{eq:knapsack-lp}.}
    \KwOut{a subset $M \subseteq E$.}
    \ForEach{$u \in U$ independently}{
        Let $\mathcal D_u$ be a distribution of $S \subseteq \delta(u)$ with $d(S) \leq b(u)$ and $\Pr[e \in S] = x_e$ for $e \in \delta(u)$. \\
        Sample $S_u \sim \mathcal D_u$.
    }
    \ForEach{$v \in W$}{
        Let $R_v := \{ e = uv \in \delta(v) : e \in S_u \}$. \\
        Find a maximum weight $T_v \subseteq R_v$ such that $d(T_v) \leq b(v)$.
    }
    \Return{$M := \bigcup_{v \in W} T_v$}
\end{algorithm}

We analyze \Cref{alg:bip-simple} on an input $(G = (U \sqcup W, E), b, d, w, x)$ where $G$ is simple and bipartite. The feasibility of the output follows from the same argument as in the proof of \Cref{lem:simple-feasible}.

\begin{lemma} \label{lem:bip-simple-feasible}
    Let $M \subseteq E$ be the output of \Cref{alg:bip-simple}. Then $d(M \cap \delta(v)) \leq b(v)$ for all $v \in U \sqcup W$.
\end{lemma}

We now analyze the weight of the output.

\begin{lemma} \label{lem:bip-simple-weight}
    Let $M \subseteq E$ be the output of \Cref{alg:bip-simple}. Then
    $$ \EE[w(M)] \geq \frac{1}{2} \cdot w^{\T} x. $$
\end{lemma}

\begin{proof}
    Let $v \in W$. Since $x$ is a feasible solution of \eqref{eq:knapsack-lp}, we have $x \in K(v)$. Each edge $e = uv \in \delta(v)$ belongs to $S_u$ with probability $x_e$. Since $G$ is simple, distinct edges incident to $v$ correspond to distinct neighbors. Since the corresponding sets $S_u$ are sampled independently across sources, each edge $e \in \delta(v)$ is independently included in $R_v$ with probability $x_e$. In other words, the random set $R_v$ is distributed as $R(x|_{\delta(v)})$. Since $T_v$ is a maximum weight subset of $R_v$ with $d(T_v) \leq b(v)$, \Cref{thm:knapsack} (with $q = 1$) and the contention-resolution characterization in \cref{sec:intro} imply that
    $$ \EE[w(T_v)] \geq \frac{1}{2} \sum_{e \in \delta(v)} w(e) x_e. $$

    Therefore,
    $$ \EE[w(M)] = \EE\left[\sum_{v \in W} w(T_v)\right] = \sum_{v \in W} \EE[w(T_v)] \geq \sum_{v \in W} \frac{1}{2} \sum_{e \in \delta(v)} w(e) x_e = \frac{1}{2} \sum_{e \in E} w(e) x_e. $$
    This completes the proof.
\end{proof}

\Cref{lem:bip-simple-feasible,lem:bip-simple-weight} imply that the integrality gap of \eqref{eq:knapsack-lp} on bipartite simple graphs is at most $2$.

\subsection{Polynomial-Time Implementation} \label{subsec:polytime}

We now show how to obtain polynomial-time approximation algorithms from \Cref{alg:simple,alg:bip-simple}. Three steps in the overall implementation cannot be performed exactly in polynomial time:
\begin{enumerate}[label=(\alph*), itemsep=0pt]
    \item computing an optimal solution of \eqref{eq:knapsack-lp}; \label{item:opt}
    \item decomposing $x|_{\delta(u)}$ into feasible knapsack solutions at each source $u$; \label{item:decompose}
    \item computing a maximum weight feasible subset of $R_v$ at each sink $v$. \label{item:knapsack-fptas}
\end{enumerate}

For \ref{item:opt}, we use the FPTAS in \Cref{lem:fptas} for a multiple-choice generalization of \eqref{eq:knapsack-lp}, whose formulation coincides with \eqref{eq:knapsack-lp} on simple graphs. For every $\gamma > 0$, we can compute, in time polynomial in $1/\gamma$ and the input length, a feasible solution $x$ of \eqref{eq:knapsack-lp} whose objective value is at least $(1 - \gamma)$ times the LP optimum. A PTAS for solving \eqref{eq:knapsack-lp} was previous obtained by \citet{jozefiak2023knapsack}.

For \ref{item:decompose}, for each source $u$, we apply the decomposition technique of \citet{lavi2011truthful} (which strengthens the decomposition technique of \citet{carr2002randomized}), together with an FPTAS for the \textsc{$0$-$1$ Knapsack} problem \citep[e.g.,][]{ibarra1975fast,lawler1979fast,magazine1981fully}, to obtain, in time polynomial in $1/\eta$ and the input length, a polynomial-size distribution $\mathcal D_u$ supported on subsets $S_u^1, \ldots, S_u^{m_u} \subseteq \delta(u)$ such that
\begin{align*}
    d(S_u^i) \leq b(u) && \forall i \in [m_u],
\end{align*}
and
\begin{align*}
    \Pr_{S_u \sim \mathcal D_u}[e \in S_u] = x_e(1 - \eta) && \forall e \in \delta(u),
\end{align*}
for some $\eta > 0$ to be determined later.

For \ref{item:knapsack-fptas}, for each sink $v$, we apply an FPTAS for the \textsc{$0$-$1$ Knapsack} problem to find a subset $T_v \subseteq R_v$ such that $d(T_v) \leq b(v)$ and
$$ w(T_v) \geq (1 - \zeta) \cdot \max\{ w(T) : T \subseteq R_v, d(T) \leq b(v) \}, $$
for some $\zeta > 0$ to be determined later.

Incorporating the modifications of \ref{item:opt}, \ref{item:decompose} and \ref{item:knapsack-fptas} into \Cref{alg:simple}, and applying the analysis of \Cref{lem:simple-weight} to $(1 - \eta) x$, we obtain
$$ \EE[w(M)] \geq w^{\T} x \cdot \frac{2q(1-q)}{1+q} \cdot (1 - \eta) \cdot (1 - \zeta) \geq \LPOPT \cdot \frac{2q(1-q)}{1+q} \cdot (1 - \eta) \cdot (1 - \zeta) \cdot (1 - \gamma), $$
where $M \subseteq E$ is the output of the modified algorithm, and $\LPOPT$ denotes the optimum of \eqref{eq:knapsack-lp}. Setting $q = \sqrt{2} - 1$ and choosing $\gamma, \eta, \zeta$ to be sufficiently small constant multiples of $\varepsilon$ gives a fully polynomial-time $(3/2 + \sqrt{2} + \varepsilon)$-approximation algorithm for simple graphs.\footnote{In the finite precision implementation, we use a sufficiently accurate rational approximation to $q = \sqrt{2} - 1$ and absorb the resulting loss into the $\varepsilon$ term. The integrality gap analysis uses the exact value of $q$.} Similarly, the same modifications applied to \Cref{alg:bip-simple} yield a fully polynomial-time $(2 + \varepsilon)$-approximation algorithm for bipartite simple graphs.

%% file: multi.tex
\section{Rounding Algorithms for Multigraphs} \label{sec:multi}

In this section, we extend the rounding algorithms from \cref{sec:simple} to multigraphs and complete the proof of \Cref{thm:main}. The main difficulty is that the independence property used for simple graphs no longer holds. Indeed, if there are parallel edges between a pair of vertices $u$ and $v$, these edges are sampled jointly as part of the same random feasible set $S_u$, and may therefore be correlated. In other words, conditioned on a vertex $v$ being a sink, the edges incident to $v$ need not be sampled independently.

We overcome this difficulty as follows. First, for each pair of vertices $u$ and $v$, we bundle each feasible nonempty subset of parallel edges into an auxiliary edge. Hence, sampling correlated parallel edges becomes choosing at most one bundle among alternative feasible bundles for each pair. Moreover, conditioned on a vertex being a sink, the bundles it receives from different neighbors are independent. However, such an auxiliary graph may have exponentially many auxiliary edges. To obtain a polynomial-size representation, we use the notion of $\xi$-Pareto families to keep only polynomially many bundles for each pair, while allowing an arbitrarily small loss in the objective. Furthermore, we prove a multiple-choice extension of the contention resolution guarantee for the integral knapsack polytope from \cref{sec:knapsack}, and use it to perform essentially the same rounding procedure as in the simple graph case.

\subsection{Bundling Parallel Edges and Pareto Families} \label{subsec:pareto}

Let $(G = (V, E), b, d, w)$ be an instance of the \textsc{Demand Matching} problem. For each $u, v \in V$, we use $F_{u, v} \subseteq E$ to denote the subset of edges with endpoints $u$ and $v$, and define
$$ \mathcal F_{u, v} := \left\{ F \subseteq F_{u, v} : F \neq \emptyset, d(F) \leq \min\{ b(u), b(v) \} \right\}. $$
We define an auxiliary graph $\mathcal G = (V, \mathcal E)$, where
$$ \mathcal E := \bigcup_{u, v \in V} \mathcal F_{u, v}. $$
Each auxiliary edge $F \in \mathcal E$ represents a nonempty subset of parallel edges of $G$ that is feasible at both endpoints, which we call a \emph{bundle}, with demand and weight
\begin{gather*}
    d(F) := \sum_{e \in F} d(e), \qquad w(F) := \sum_{e \in F} w(e).
\end{gather*}
Then the original \textsc{Demand Matching} problem on $(G, b, d, w)$ is equivalent to finding a demand matching in $\mathcal G$ subject to the additional requirement that at most one auxiliary edge is selected between each pair of vertices. \Cref{fig:bundling} illustrates the construction of the auxiliary graph.

\begin{figure}[ht]
\centering

\begin{subfigure}[t]{0.475\textwidth}
\centering
\begin{tikzpicture}[thick]
    \node[circle, draw, minimum size=8mm] (u) at (0,0) {$u$};
    \node[circle, draw, minimum size=8mm] (v) at (3,0) {$v$};

    \draw[blue, very thick, bend left=35]  (u) to node[above, blue] {$e_1$} (v);
    \draw[blue, very thick]                (u) to node[above, blue] {$e_2$} (v);
    \draw[bend right=35]                   (u) to node[below] {$e_3$} (v);
\end{tikzpicture}
\caption{Original parallel edges. Here, \(b(u)=b(v)=5\), and
\(d(e_1)=2\), \(d(e_2)=3\), and \(d(e_3)=4\). The highlighted edges \(e_1\) and \(e_2\) form a feasible subset.}
\label{fig:bundle-original}
\end{subfigure}
\hfill
\begin{subfigure}[t]{0.475\textwidth}
\centering
\begin{tikzpicture}[thick]
    \node[circle, draw, minimum size=8mm] (u) at (0,0) {$u$};
    \node[circle, draw, minimum size=8mm] (v) at (4,0) {$v$};

    \draw[bend left=55]  (u) to node[above, font=\scriptsize] {$\{e_1\}$} (v);
    \draw[bend left=18]  (u) to node[above, font=\scriptsize] {$\{e_2\}$} (v);
    \draw[bend right=18] (u) to node[above, font=\scriptsize] {$\{e_3\}$} (v);
    \draw[blue, very thick, bend right=55] (u) to node[above, font=\scriptsize, blue] {$\{e_1,e_2\}$} (v);
\end{tikzpicture}
\caption{Auxiliary edges corresponding to the feasible nonempty subsets of \(F_{u,v}\). Their demands are \(2\), \(3\), \(4\), and \(5\), respectively. The highlighted auxiliary edge represents the feasible subset \(\{e_1,e_2\}\).}
\label{fig:bundle-aux}
\end{subfigure}

\caption{Bundling parallel edges between \(u\) and \(v\). Each feasible subset of parallel edges in the original graph is represented by one auxiliary edge; at most one such auxiliary edge may be selected between \(u\) and \(v\).}
\label{fig:bundling}
\end{figure}
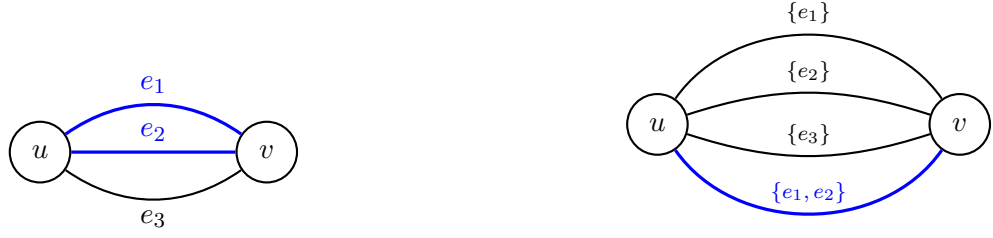

However, the auxiliary graph $\mathcal G$ may have exponentially many edges. To obtain a polynomial-size representation, we replace each family $\mathcal F_{u, v}$ with a polynomial-size approximate representation, which we call a $\xi$-Pareto family. Formally, given a family $\mathcal F \subseteq 2^N$ on a finite ground set $N$, $d : N \to \NN$, $w : N \to \ZZ_+$, and $\xi > 0$, a subfamily $\mathcal F^\xi \subseteq \mathcal F$ is a \emph{$\xi$-Pareto family} of $\mathcal F$ if, for every $S \in \mathcal F$, there exists $S' \in \mathcal F^\xi$ such that
\begin{gather*}
    d(S') \leq d(S), \qquad w(S') \geq (1 - \xi) \cdot w(S).
\end{gather*}
In other words, each subset $S \in \mathcal F$ can be replaced by a representative $S'$ in the $\xi$-Pareto family $\mathcal F^\xi$ with no larger demand and with at most a $\xi$-fractional loss in weight.

\citet{herzel2021one} showed that, for any $\mathcal F \subseteq 2^N$ and $\xi > 0$, there exists a $\xi$-Pareto family of $\mathcal F$ whose cardinality is polynomial in $1/\xi$ and the maximum encoding length of $d(S)$ and $w(S)$ over $S \in \mathcal F$. Moreover, they characterized when such a $\xi$-Pareto family can be constructed efficiently: a polynomial-size $\xi$-Pareto family can be computed in polynomial time if and only if an associated auxiliary problem can be solved in polynomial time.\footnote{Unless noted otherwise, polynomial bounds are with respect to $1/\xi$ and the maximum encoding length of $d(S)$ and $w(S)$ over $S \in \mathcal F$ when referring to a $\xi$-Pareto family of $\mathcal F$.} We show that this auxiliary problem is polynomial-time solvable in our setting, and hence that a polynomial-size $\xi$-Pareto family can be computed in polynomial time for any family of bundles. We defer the proof to \cref{apx:pareto}.

\begin{lemma} \label{lem:pareto}
    Let $N$ be a finite ground set. Let $d : N \to \NN$ and $w : N \to \ZZ_+$. Let $B \in \ZZ_+$. Let
    $$ \mathcal F := \{ S \subseteq N : S \neq \emptyset, d(S) \leq B \}. $$
    Let $\xi > 0$. Then a $\xi$-Pareto family of $\mathcal F$ of polynomial cardinality can be computed in polynomial time, where the polynomial bounds are with respect to $1/\xi$ and the encoding length of $(N, d, w, B)$.
\end{lemma}

Fix $\xi > 0$. We define a \emph{$\xi$-Pareto auxiliary graph} $\mathcal G^\xi = (V, \mathcal E^\xi)$. For each unordered pair of vertices $u, v \in V$, let $\mathcal F_{u, v}^\xi$ be a polynomial-size $\xi$-Pareto family of $\mathcal F_{u, v}$. Define\footnote{By slight abuse of notation, we use the indexing notation $u, v \in V$ to denote an unordered pair of vertices.}
$$ \mathcal E^\xi := \bigcup_{u, v \in V} \mathcal F_{u, v}^\xi. $$
For each $v \in V$, we define\footnote{Given a graph $G = (V, E)$ and $v \in V$, we use $N(v) \subseteq V$ to denote the set of neighbors of $v$ in $G$.}
$$ \delta^\xi(v) := \bigcup_{u \in N(v)} \mathcal F_{u, v}^\xi. $$
We say that a subset $\mathcal M \subseteq \mathcal E^\xi$ is a \emph{multiple-choice demand matching} in $\mathcal G^\xi$ if
\begin{align*}
    d\left(\mathcal M \cap \delta^\xi(v)\right) \leq b(v) && \forall v \in V,
\end{align*}
and
\begin{align*}
    \left|\mathcal M \cap \mathcal F_{u, v}^\xi\right| \leq 1 && \forall u, v \in V.
\end{align*}
Each multiple-choice demand matching $\mathcal M$ in $\mathcal G^\xi$ corresponds to a feasible demand matching $M$ in the original \textsc{Demand Matching} instance $(G, b, d, w)$, obtained by replacing each auxiliary edge $F \in \mathcal M$ with the set of original edges it represents, i.e.,
$$ M := \bigcup_{F \in \mathcal M} F. $$
Conversely, take an optimal demand matching $M^\star$ in $G$. For each unordered pair of vertices $u, v \in V$ with $F_{u, v}^\star := M^\star \cap F_{u, v} \neq \emptyset$, choose a $\xi$-Pareto representative $F_{u, v}' \in \mathcal F_{u, v}^\xi$ for $F_{u, v}^\star$, which satisfies $d(F_{u, v}') \leq d(F_{u, v}^\star)$ and $w(F_{u, v}') \geq (1 - \xi) \cdot w(F_{u, v}^\star)$. Hence,
$$ \mathcal M' := \left\{ F_{u, v}' : u, v \in V, F_{u, v}^\star \neq \emptyset \right\} $$
is a feasible multiple-choice demand matching in $\mathcal G^\xi$ with
$$ w(\mathcal M') \geq (1 - \xi) \cdot w(M^\star). $$
It follows that a maximum weight multiple-choice demand matching in $\mathcal G^\xi$ has weight at least $(1 - \xi)$ times the optimum of the original instance. In the next subsection, we define the \textsc{Multiple-Choice Demand Matching} problem more generally.

\subsection{Multiple-Choice Demand Matching}

The preceding construction motivates the following generalization of the \textsc{Demand Matching} problem. Let $G = (V, E)$ be a (multi)graph with \emph{vertex capacities} $b : V \to \ZZ_+$, edge demands $d : E \to \NN$, and edge weights $w : E \to \ZZ_+$. For each $u, v \in V$, let $F_{u, v} \subseteq E$ be the subset of edges with endpoints $u$ and $v$. Then $\Pi_v := \{ F_{u, v} : u \in N(v) \}$ is a partition of $\delta(v)$ for each $v \in V$. We say that a subset $M \subseteq E$ is a \emph{multiple-choice demand matching} if
\begin{align*}
    d(M \cap \delta(v)) \leq b(v) && \forall v \in V,
\end{align*}
and
\begin{align*}
    |M \cap F_{u, v}| \leq 1 && \forall u, v \in V.
\end{align*}
The \textsc{Multiple-Choice Demand Matching} problem is to find a multiple-choice demand matching $M$ such that $w(M)$ is maximized. Without loss of generality, we assume that $G$ does not contain self-loops. The $\xi$-Pareto auxiliary instance defined in the preceding subsection is a special case of the \textsc{Multiple-Choice Demand Matching} problem, where the input graph is $G^\xi$.

For each $v \in V$, we define
\begin{align*}
    \mathcal K_{\mathsf{mc}}(v) &:= \left\{ S \subseteq E : d(S \cap \delta(v)) \leq b(v), |S \cap \pi| \leq 1 \;\forall \pi \in \Pi_v \right\}, \\
    K_{\mathsf{mc}}(v) &:= \Conv\left\{ \chi^S : S \in \mathcal K_{\mathsf{mc}}(v) \right\}.
\end{align*}
Analogous to the knapsack intersection LP, we define the following LP relaxation of the \textsc{Multiple-Choice Demand Matching} problem, which we call the \emph{multiple-choice knapsack intersection LP}:
\begin{alignat}{4}
    \maximize \qquad && \sum_{e \in E} x_e w(e) & \tag{MC-Knapsack-LP} \label{eq:mc-knapsack-lp} \\
    \subjectto \qquad && x \in K_{\mathsf{mc}}(v) && \qquad \forall v \in V. \notag
\end{alignat}
We show that there exists an FPTAS for optimization over \eqref{eq:mc-knapsack-lp}. We defer its proof to \cref{apx:fptas}.

\begin{lemma} \label{lem:fptas}
    For every $\varepsilon > 0$, one can compute, in time polynomial in $1/\varepsilon$ and the input length, a feasible solution $x$ of \eqref{eq:mc-knapsack-lp} such that $\sum_{e \in E} x_e w(e)$ is at least $(1 - \varepsilon)$ times the optimum of \eqref{eq:mc-knapsack-lp}. In other words, there exists an FPTAS for solving \eqref{eq:mc-knapsack-lp}.
\end{lemma}

\subsection{Multiple-Choice CR Scheme for the Integral Knapsack Polytope} \label{subsec:mc-knapsack}

We now extend the contention resolution result from \cref{sec:knapsack} to the multiple-choice sampling structure arising in the \textsc{Multiple-Choice Demand Matching} problem. Recall that we have a partition of $\delta(v)$ for each vertex $v \in V$, and the output $M$ chooses at most one edge from each part in the partition. Therefore, we define the following multiple-choice sampling structure for this setting. Let $N$ be a finite ground set. Let $\Pi$ be a partition of $N$. Instead of sampling each item independently, we independently process each part $\pi \in \Pi$. Formally, given $y \in [0, 1]^N$ such that $y(\pi) \leq 1$ for all $\pi \in \Pi$, let $R_\Pi(y) \subseteq N$ denote the random subset obtained by, independently for each $\pi \in \Pi$, choosing either one item $e \in \pi$ with probability $y_e$, or no item in $\pi$ with probability $1 - y(\pi)$. When every part of $\Pi$ is a singleton, this reduces to the independent sampling process $R(y)$ considered in \cref{sec:knapsack}.

We show that the same expected weight guarantee in \Cref{thm:knapsack} continues to hold under this multiple-choice sampling structure. We remark that $x$ in the statement of the theorem below only needs to belong to the integral knapsack polytope. The multiple-choice restriction applies to the sampled set $R_\Pi(qx)$, while the feasible sets in a convex decomposition of $x$ may contain several items from the same part of $\Pi$.

\begin{theorem} \label{thm:mc-knapsack}
    Let $d_1, \ldots, d_n \in \NN$. Let $B \in \ZZ_+$. Let $w \in \RR_+^n$. Let $\Pi$ be a partition of $[n]$. Let $x \in P_{\mathcal K}^{\mathsf{int}}$ and $q \in [0, 1]$ be such that $qx(\pi) \leq 1$ for all $\pi \in \Pi$. Then
    $$ \EE[\max\{ w(S) : S \subseteq R_\Pi(qx), d(S) \leq B \}] \geq \frac{q}{1 + q} \cdot w^{\T} x. $$
\end{theorem}

The proof follows from the same argument using deletion sensitivity and an insertion lemma as in \cref{sec:knapsack}. Deletion monotonicity (\Cref{lem:mono} and \Cref{cor:mono}) and bounded total deletion sensitivity (\Cref{lem:del}) are not dependent on the sampling structure and therefore continue to apply unchanged in the multiple-choice setting. The only issue from the multiple-choice sampling structure is that inserting an item $e$ may conflict with an item already selected from the same part of $\Pi$. Therefore, we first remove all items from that part, if any, and then insert $e$. Formally, given $e \in [n]$, we use $\pi(e)$ to denote the part of $\Pi$ containing $e$, and replace \eqref{eq:insert} with
\begin{equation} \label{eq:mc-insert}
    T \vee e := (T \setminus \pi(e)) \cup \{ e \}
\end{equation}
for each $T \subseteq [n]$ and $e \in [n]$. We have the following multiple-choice insertion lemma.

\begin{lemma}[Multiple-choice insertion lemma] \label{lem:mc-insertion}
    Let $S, T \subseteq [n]$ be such that $d(S) \leq B$. Then
    $$ w(S) \leq r(T, B) + \sum_{e \in S} \Delta_e(T \vee e, B). $$
\end{lemma}

The proof of \Cref{lem:mc-insertion} is identical to that of \Cref{lem:insertion}, after replacing \eqref{eq:insert} with \eqref{eq:mc-insert}, and is omitted for conciseness.

With \Cref{cor:mono} and \Cref{lem:del,lem:mc-insertion}, the proof of \Cref{thm:mc-knapsack} follows essentially the same argument as that of \Cref{thm:knapsack}. The only additional observations are that $R \vee e = R$ for each $e \in R$, and that $\Delta_e(R \vee e, B)$ is independent of the event $e \in R$, where $R = R_\Pi(qx)$. We omit the proof for conciseness.

\subsection{Rounding Algorithms for Multiple-Choice Demand Matching}

We now extend the algorithms from \cref{sec:simple} to the \textsc{Multiple-Choice Demand Matching} problem. The algorithms follow the same source-sink rounding framework as before. The main difference is that, for each pair of vertices $u$ and $v$, at most one edge may be selected from $F_{u, v}$. Therefore, the edges proposed to a sink are no longer sampled independently; instead, they follow the multiple-choice sampling structure defined in \cref{subsec:mc-knapsack}. Since multiple-choice contention resolution for the integral knapsack polytope has the same guarantee as (standard) contention resolution, we have the same approximation analysis. We formally state the algorithms for general and bipartite graphs in \Cref{alg:mc,alg:mc-bip}, respectively.

\begin{algorithm}[ht]
    \caption{A rounding algorithm for \textsc{Multiple-Choice Demand Matching}.}
    \label{alg:mc}
    \KwIn{a (multi)graph $G = (V, E)$, $b : V \to \ZZ_+$, $d : E \to \NN$, $w : E \to \ZZ_+$, $q \in [0, 1]$, and a feasible solution $x$ of \eqref{eq:mc-knapsack-lp}.}
    \KwOut{a subset $M \subseteq E$.}
    \ForEach{$v \in V$ independently}{
        Let $\mathcal D_v$ be a distribution over subsets $S \subseteq \delta(v)$ with $d(S) \leq b(v)$, $|S \cap \pi| \leq 1$ for all $\pi \in \Pi_v$, and $\Pr[e \in S] = x_e$ for all $e \in \delta(v)$. \\
        Label $v$ as a \emph{source} with probability $q$ and as a \emph{sink} otherwise.
    }
    \ForEach{source $u \in V$ independently}{
        Sample $S_u \sim \mathcal D_u$ independent of the vertex labels.
    }
    \ForEach{sink $v \in V$}{
        Let $R_v := \bigcup \{ S_u \cap F_{u, v} : u \in N(v) \text{ is a source} \}$. \\
        Find a maximum weight $T_v \subseteq R_v$ such that $d(T_v) \leq b(v)$.
    }
    \Return{$M := \bigcup \{ T_v : v \in V \text{ sink} \}$}
\end{algorithm}

\begin{algorithm}[ht]
    \caption{A rounding algorithm for \textsc{Multiple-Choice Demand Matching} on bipartite (multi)graphs.}
    \label{alg:mc-bip}
    \KwIn{a bipartite (multi)graph $G = (U \sqcup W, E)$, $b : U \sqcup W \to \ZZ_+$, $d : E \to \NN$, $w : E \to \ZZ_+$, and a feasible solution $x$ of \eqref{eq:mc-knapsack-lp}.}
    \KwOut{a subset $M \subseteq E$.}
    \ForEach{$u \in U$ independently}{
        Let $\mathcal D_u$ be a distribution over subsets $S \subseteq \delta(u)$ with $d(S) \leq b(u)$, $|S \cap \pi| \leq 1$ for all $\pi \in \Pi_u$, and $\Pr[e \in S] = x_e$ for all $e \in \delta(u)$. \\
        Sample $S_u \sim \mathcal D_u$.
    }
    \ForEach{$v \in W$}{
        Let $R_v := \bigcup \{ S_u \cap F_{u, v} : u \in N(v) \}$. \\
        Find a maximum weight $T_v \subseteq R_v$ such that $d(T_v) \leq b(v)$.
    }
    \Return{$M := \bigcup_{v \in W} T_v$}
\end{algorithm}

Let $M \subseteq E$ be the output of \Cref{alg:mc} or \Cref{alg:mc-bip}. The feasibility of vertex-capacity constraints follows exactly as in \Cref{lem:simple-feasible,lem:bip-simple-feasible}. It remains to verify the multiple-choice constraints. Let $u, v \in V$. If $u$ and $v$ are both sources or both sinks, then $M \cap F_{u, v} = \emptyset$. Otherwise, suppose without loss of generality that $u$ is a source and $v$ is a sink. Since $S_u$ is feasible in $\mathcal K_{\mathsf{mc}}(u)$ and since $M \cap F_{u, v} \subseteq S_u \cap F_{u, v}$,
$$ \left|M \cap F_{u, v}\right| \leq \left|S_u \cap F_{u, v}\right| \leq 1. $$
Hence, $M$ is a multiple-choice demand matching. We state the feasibility of \Cref{alg:mc,alg:mc-bip} in the following lemmas.

\begin{lemma} \label{lem:mc-feasible}
    Let $M \subseteq E$ be the output of \Cref{alg:mc}. Then $d(M \cap \delta(v)) \leq b(v)$ for all $v \in V$, and $|M \cap F_{u, v}| \leq 1$ for all $u, v \in V$.
\end{lemma}

\begin{lemma} \label{lem:mc-bip-feasible}
    Let $M \subseteq E$ be the output of \Cref{alg:mc-bip}. Then $d(M \cap \delta(v)) \leq b(v)$ for all $v \in U \sqcup W$, and $|M \cap F_{u, v}| \leq 1$ for all $u, v \in U \sqcup W$.
\end{lemma}

We next analyze the expected weight of the output. The argument follows from the same analysis as in \Cref{lem:simple-weight,lem:bip-simple-weight}, with \Cref{thm:mc-knapsack} replacing \Cref{thm:knapsack}. In particular, we verify that the assumptions of \Cref{thm:mc-knapsack} are satisfied.

First consider \Cref{alg:mc}. Assume without loss of generality that $0 < q < 1$. Let $v \in V$. Since $x \in K_{\mathsf{mc}}(v) \subseteq K(v)$, we have $x \in K(v)$ and $x(F_{u, v}) \leq 1$ for each $u \in N(v)$. For each $u \in N(v)$ and $e \in F_{u, v}$,
$$ \Pr\left[R_v \cap F_{u, v} = \{ e \} \;\middle|\; v \text{ is a sink}\right] = \Pr[u \text{ is a source}] \cdot \Pr\left[e \in S_u \;\middle|\; u \text{ is a source}\right] = qx_e. $$
Hence, for each $u \in N(v)$,
$$ \Pr\left[R_v \cap F_{u, v} = \emptyset \;\middle|\; v \text{ is a sink}\right] = 1 - qx(F_{u, v}). $$
Moreover, conditioned on $v$ being a sink, the random sets $R_v \cap F_{u, v}$ are independent across different neighbors $u \in N(v)$, and hence has the multiple-choice sampling structure of \Cref{thm:mc-knapsack} with partition $\Pi_v$; in other words, $R_v$ is distributed as $R_{\Pi_v}(qx|_{\delta(v)})$. Hence, \Cref{thm:mc-knapsack} implies that
$$ \EE\left[w(T_v) \;\middle|\; v \text{ is a sink}\right] \geq \frac{q}{1 + q} \sum_{e \in \delta(v)} w(e) x_e. $$
For \Cref{alg:mc-bip}, the same argument yields that for each sink $v$,
$$ \EE\left[w(T_v)\right] \geq \frac{1}{2} \sum_{e \in \delta(v)} w(e) x_e. $$
The remainder of the analysis is identical to \Cref{lem:simple-weight,lem:bip-simple-weight} by summing over all vertices (for \Cref{alg:mc}) or over all sinks (for \Cref{alg:mc-bip}). We state the expected weight guarantees of \Cref{alg:mc,alg:mc-bip} in the following lemmas.

\begin{lemma} \label{lem:mc-weight}
    Let $M \subseteq E$ be the output of \Cref{alg:mc}. Then
    $$ \EE[w(M)] \geq w^{\T} x \cdot \frac{2q(1-q)}{1+q}. $$
    Choosing $q = \sqrt{2} - 1$ gives that $\EE[w(M)] \geq (6 - 4\sqrt{2}) \cdot w^{\T} x$.
\end{lemma}

\begin{lemma} \label{lem:mc-bip-weight}
    Let $M \subseteq E$ be the output of \Cref{alg:mc-bip}. Then
    $$ \EE[w(M)] \geq \frac{1}{2} \cdot w^{\T} x. $$
\end{lemma}

\paragraph{Polynomial-Time Implementation.} We briefly discuss the polynomial-time implementation of \Cref{alg:mc,alg:mc-bip}. As in \cref{subsec:polytime}, three steps in the overall implementation cannot be performed exactly in polynomial time, which we address in essentially the same way below. First, \Cref{lem:fptas} implies that, for every $\gamma > 0$, we can compute, in time polynomial in $1/\gamma$ and the input length, a feasible solution $x$ of \eqref{eq:mc-knapsack-lp} whose objective value is at least $(1 - \gamma)$ times the LP optimum. Second, for each source $u$, using the decomposition technique of \citet{lavi2011truthful} together with an FPTAS for the \textsc{$0$-$1$ Multiple-Choice Knapsack} problem \citep[e.g.,][]{lawler1979fast,bansal2004improved}, we can compute, in time polynomial in $1/\eta$ and the input length, a polynomial-support distribution over sets that are feasible in $\mathcal K_{\mathsf{mc}}(u)$ and that satisfy
\begin{align*}
    \Pr[e \in S_u] = x_e (1 - \eta) && \forall e \in \delta(u),
\end{align*}
for some $\eta > 0$ to be determined later. Finally, for each sink $v$, we apply an FPTAS for the \textsc{$0$-$1$ Knapsack} problem \citep[e.g.,][]{ibarra1975fast,lawler1979fast,magazine1981fully} to find a subset $T_v \subseteq R_v$ such that $d(T_v) \leq b(v)$ and
$$ w(T_v) \geq (1 - \zeta) \cdot \max\{ w(T) : T \subseteq R_v, d(T) \leq b(v) \}, $$
for some $\zeta > 0$ to be determined later. By setting $q = \sqrt{2} - 1$ (for \Cref{alg:mc}) and choosing $\gamma, \eta, \zeta$ to be sufficiently small constant multiples of $\varepsilon$, we obtain fully polynomial-time approximation algorithms with approximation ratios of $3/2 + \sqrt{2} + \varepsilon$ and $2 + \varepsilon$ for general and bipartite instances of the \textsc{Multiple-Choice Demand Matching} problem, respectively, for every $\varepsilon > 0$.

\subsection{Putting Everything Together}

We now combine the preceding ingredients to prove \Cref{thm:main}. Let $(G = (V, E), b, d, w)$ be an instance of the \textsc{Demand Matching} problem. Let $\OPT$ be the optimum of this instance. Fix $\xi > 0$. As described in \cref{subsec:pareto}, we construct the $\xi$-Pareto auxiliary graph $\mathcal G^\xi = (V, \mathcal E^\xi)$, which defines an instance of the \textsc{Multiple-Choice Demand Matching} problem. This auxiliary instance has size polynomial in $1/\xi$ and the input length. Let $\OPT^\xi$ denote the optimum of this auxiliary instance. Each multiple-choice demand matching $\mathcal M$ in $\mathcal G^\xi$ expands to a demand matching
$$ M := \bigcup_{F \in \mathcal M} F $$
in $G$ with $w(M) = w(\mathcal M)$. By the $\xi$-Pareto property,
$$ \OPT^\xi \geq (1 - \xi) \cdot \OPT. $$

For general graphs, we apply the $(3/2 + \sqrt{2} + \rho)$-approximation algorithm for the \textsc{Multiple-Choice Demand Matching} problem to $\mathcal G^\xi$ for some $\rho > 0$ to be determined later. This produces a multiple-choice demand matching $\mathcal M$ satisfying
$$ \EE[w(\mathcal M)] \geq \frac{\OPT^\xi}{\frac{3}{2} + \sqrt{2} + \rho} \geq \frac{1 - \xi}{\frac{3}{2} + \sqrt{2} + \rho} \cdot \OPT. $$
Choosing $\xi$ and $\rho$ to be sufficiently small constant multiples of $\varepsilon$ yields a randomized $(3/2 + \sqrt{2} + \varepsilon)$-approximation algorithm for the \textsc{Demand Matching} problem.

For bipartite graphs, we note that $\mathcal G^\xi$ is also bipartite since every bundle has the same endpoints as its original parallel edges, so we apply the $(2 + \rho)$-approximation algorithm for the \textsc{Multiple-Choice Demand Matching} problem to $\mathcal G^\xi$ for some $\rho > 0$ to be determined later. This produces a multiple-choice demand matching $\mathcal M$ satisfying
$$ \EE[w(\mathcal M)] \geq \frac{\OPT^\xi}{2 + \rho} \geq \frac{1 - \xi}{2 + \rho} \cdot \OPT. $$
Choosing $\xi$ and $\rho$ to be sufficiently small constant multiples of $\varepsilon$ yields a randomized $(2 + \varepsilon)$-approximation algorithm for the \textsc{Demand Matching} problem.

Since the $\xi$-Pareto auxiliary instance has size polynomial in $1/\xi$ and the input length, both algorithms run in fully polynomial time. This proves \Cref{thm:main}.

%% file: concluding.tex
\section{Concluding Remarks} \label{sec:concluding}

In this paper, we have combined knapsack intersection relaxations with contention resolution to obtain fully polynomial-time randomized approximation algorithms for the \textsc{Demand Matching} problem, achieving ratios of $3/2 + \sqrt{2} + \varepsilon \approx 2.914 + \varepsilon$ for general graphs and $2 + \varepsilon$ for bipartite graphs in expectation, for every $\varepsilon > 0$. This gives the first approximation ratio strictly below $3$, overcoming the integrality gap barrier of the natural LP relaxation through a stronger formulation. A key ingredient is a $(q, 1/(1+q))$-balanced contention resolution scheme for the integral knapsack polytope for every $q \in [0, 1]$, which may be of independent interest in other packing problems with local knapsack constraints. The balance guarantee $1/(1+q)$ is optimal in the worst case over all knapsack instances.

Several questions remain open. First, can the approximation ratios be improved further for general or bipartite graphs? In particular, determining the exact integrality gaps of \eqref{eq:knapsack-lp} on general and bipartite graphs remains an interesting question. More broadly, \eqref{eq:knapsack-lp} is the first level of the knapsack intersection hierarchy introduced by \citet{jozefiak2023knapsack}. Can a fixed higher level of this hierarchy yield approximation algorithms with better guarantees? Another direction is to identify further applications of the knapsack contention resolution guarantee and its multiple-choice extension. Finally, we ask whether the same approximation guarantees for multigraphs can be achieved by efficiently rounding \eqref{eq:knapsack-lp} directly, without the auxiliary bundle construction.

%% file: acknowledgements.tex
\paragraph{Acknowledgements.} We thank Bruce Shepherd for fruitful discussions. We used GPT-6 Astra to assist with editing and checking for mathematical correctness. In addition to minor tweaks, the tool identified a fixable mistake in \cref{apx:tight}. The tool was also used for the creation of \Cref{tab:dm} and \Cref{fig:bundling}. The authors produced the mathematical results in the paper themselves and vouch for their originality.

%% file: tight.tex
\section{Tightness of the Knapsack Contention Resolution Guarantee} \label{apx:tight}

In this appendix, we prove \Cref{prop:tight}. In particular, we define the fractional solution $x$ to be a convex combination of two feasible solutions, one consisting of a single large item that occupies the entire knapsack, and the other consisting of many unit-demand items. Hence, an optimal subset of the sampled set takes either the large item or the sampled unit-demand items, but not both. We choose the weights so that the weight of the large item approximates the expected total weight of the sampled unit-demand items, and that
$$ \frac{\EE[r(R, B)]}{qw^{\T} x} < \frac{1}{1 + q} + \varepsilon, $$
where $R := R(qx)$ and $r(R, B) := \max\{ w(T) : T \subseteq R, d(T) \leq B \}$. The proof below makes this intuition rigorous.

\begin{proof}[Proof of \Cref{prop:tight}]
    Let $\varepsilon > 0$. Let $M \in \NN \cap [2, \infty)$ be a parameter to be determined later. Let $B := M^2$. Let $n := M^2 + 1$. Let
    \begin{alignat*}{3}
        & d_i := 1 && \qquad w_i := 1 && \qquad \forall i \in [n - 1], \\
        & d_n := M^2 && \qquad w_n := \lfloor qM \rfloor.
    \end{alignat*}
    Recall that $P_{\mathcal K}^{\mathsf{int}} := \Conv\{ \chi^S : S \subseteq [n], d(S) \leq B \}$. Then $\chi^{[n - 1]}, \chi^{\{ n \}} \in P_{\mathcal K}^{\mathsf{int}}$. Hence,
    $$ x := \frac{1}{M} \chi^{[n - 1]} + \left(1 - \frac{1}{M}\right) \chi^{\{ n \}} \in P_{\mathcal K}^{\mathsf{int}}, $$
    where $x_i = 1/M$ for all $i \in [n - 1]$ and $x_n = 1 - 1/M$. We have
    \begin{align*}
        w^{\T} x &= (n - 1) \cdot 1 \cdot \frac{1}{M} + \lfloor qM \rfloor \cdot \left(1 - \frac{1}{M}\right) \geq M^2 \cdot \frac{1}{M} + (qM - 1) \left(1 - \frac{1}{M}\right) \\
        &= qM + M - q - 1 + \frac{1}{M} = (1 + q)M - (1 + q) + \frac{1}{M} \geq (1 + q)M - 2.
    \end{align*}
    Let $R := R(qx)$. Let $A$ be the event that $n \in R$. Then
    $$ \Pr[A] = qx_n = q\left(1 - \frac{1}{M}\right). $$
    Let $X := |R \cap [n - 1]| = |R \cap [M^2]|$. Then $X \sim \Bin(M^2, q/M)$. Hence,
    \begin{align*}
        \EE[X] &= M^2 \cdot \frac{q}{M} = qM, \\
        \Var[X] &= M^2 \cdot \frac{q}{M} \cdot \left(1 - \frac{q}{M}\right) \leq qM.
    \end{align*}
    Then
    \begin{align*}
        r(R, B) &= \left\{
            \begin{array}{ll}
                \max\left\{ w_n, X \right\} & \text{if $A$ occurs}, \\
                X & \text{otherwise}
            \end{array}
        \right. \\
        &= X + \mathds 1[A] \cdot \left(w_n - X\right)_+,
    \end{align*}
    where we use the notation $z_+ := \max\{ z, 0 \}$ for all $z \in \RR$. By the independence of $\mathds 1[A]$ and $X$,
    $$ \EE[r(R, B)] = \EE[X] + \Pr[A] \cdot \EE\left[\left(w_n - X\right)_+\right]. $$
    Since $w_n = \lfloor qM \rfloor \leq qM$, we have $(w_n - X)_+ \leq (qM - X)_+ \leq |qM - X| = |\EE[X] - X|$. By the Cauchy--Schwarz inequality,
    $$ \EE\left[\left(w_n - X\right)_+\right] \leq \EE[|\EE[X] - X|] \leq \sqrt{\EE\left[\left(\EE[X] - X\right)^2\right]} = \sqrt{\Var[X]} \leq \sqrt{qM}. $$
    Hence,
    $$ \EE[r(R, B)] \leq qM + q\left(1 - \frac{1}{M}\right) \cdot \sqrt{qM} \leq qM + q\sqrt{qM}. $$
    It follows that
    $$ \frac{\EE[r(R, B)]}{qw^{\T} x} \leq \frac{qM + q\sqrt{qM}}{q((1 + q)M - 2)} = \frac{M + \sqrt{qM}}{(1 + q)M - 2} = \frac{1 + \sqrt{\frac{q}{M}}}{1 + q - \frac{2}{M}} \to \frac{1}{1 + q}, $$
    as $M \to \infty$. Hence, for sufficiently large $M$, we have $\EE[r(R, B)]/(qw^{\T} x) < 1/(1 + q) + \varepsilon$, i.e.,
    $$ \EE[r(R, B)] < q\left(\frac{1}{1 + q} + \varepsilon\right) w^{\T} x. $$
    This completes the proof.
\end{proof}

%% file: pareto.tex
\section{Polynomial-Time Construction of a Pareto Family} \label{apx:pareto}

In this appendix, we prove \Cref{lem:pareto}. Let $N$ be a finite ground set. Let $d : N \to \NN$. Let $w : N \to \ZZ_+$. Let $B \in \ZZ_+$. Let
$$ \mathcal F := \{ S \subseteq N: S \neq \emptyset, d(S) \leq B \}. $$
Let $\xi > 0$. Unless noted otherwise, polynomial bounds in this appendix are with respect to $1/\xi$ and the encoding length of $(N, d, w, B)$.

We preprocess the input to handle several boundary cases. First, we discard every item $e \in N$ with $d(e) > B$, since such an item belongs to no feasible set. If no item remains, then $\mathcal F = \emptyset$, so $\emptyset$ is a $\xi$-Pareto family of $\mathcal F$. Otherwise, let $e_0 \in N$ be a remaining item of minimum demand, and let $N_+ := \{ e \in N : w(e) > 0 \}$. If $N_+ = \emptyset$, then $\{ \{ e_0 \} \}$ is a $\xi$-Pareto family of $\mathcal F$. Otherwise, it suffices to construct a $\xi$-Pareto family $\mathcal F^\xi$ of
$$ \mathcal F_+ := \left\{ S \subseteq N_+ : S \neq \emptyset, d(S) \leq B \right\}, $$
and $\mathcal F^\xi \cup \{ \{ e_0 \} \}$ is a $\xi$-Pareto family of $\mathcal F$. Indeed, every set $S \in \mathcal F$ with $w(S) > 0$ can be replaced by $S \cap N_+ \in \mathcal F_+$ without increasing its demand or changing its weight, while every set $S \in \mathcal F$ with $w(S) = 0$ is represented by $\{ e_0 \}$. Hence, we assume without loss of generality that $N \neq \emptyset$, and $d(e) \leq B$ and $w(e) \in \NN$ for all $e \in N$.

\citet[Theorem 4]{herzel2021one} showed that a $\xi$-Pareto family of polynomial cardinality can be computed in polynomial time if and only if the following problem, which we call the \textsc{Threshold Knapsack} problem, can be solved in time polynomial in $1/\varphi$ and the encoding length: given $\varphi \in \QQ \cap (0, 1)$ and $\Psi \in \QQ_{> 0}$, either certify that there is no $S \in \mathcal F$ with $w(S) \geq \Psi$, or return some $S \in \mathcal F$ such that
\begin{gather*}
    d(S) \leq d^\star(\Psi), \qquad w(S) \geq (1 - \varphi) \Psi,
\end{gather*}
where\footnote{We use the convention that the minimum over an empty feasible set is $\infty$.}
$$ d^\star(\Psi) := \min\{ d(S) : S \in \mathcal F, w(S) \geq \Psi \}. $$

We show that this problem admits a fully polynomial-time algorithm by dynamic programming. In particular, we round the item weights so that the total rounded weight takes only polynomially many values, while losing at most a $\varphi$-fraction of the target weight $\Psi$. We then use a standard knapsack dynamic program that finds a set of minimum demand for each rounded weight.

\begin{lemma} \label{lem:threshold}
    Given $\varphi \in \QQ \cap (0, 1)$ and $\Psi \in \QQ_{> 0}$, the \textsc{Threshold Knapsack} problem can be solved in time polynomial in $1/\varphi$ and the encoding length of $(N, d, w, B, \varphi, \Psi)$.
\end{lemma}

\begin{proof}
    Let $m := |N|$. Let $M := (\varphi \Psi)/m$. For each $e \in N$, we define $\hat w(e) := \lfloor \min\{ w(e), \Psi \}/M \rfloor$. Let $W := \lceil m(1-\varphi)/\varphi \rceil$. Suppose that $N = \{ e_1, \ldots, e_m \}$. For each $i \in \{ 0, \ldots, m \}$ and for each $t \in \{ 0, \ldots, W \}$, we define
    \begin{equation} \label{eq:dp}
        \DP(i, t) := \min\{ d(S) : S \subseteq \{ e_1, \ldots, e_i \}, \min\{ \hat w(S), W \} = t \}.
    \end{equation}
    We show that all entries defined in \eqref{eq:dp} for all $i \in \{ 0, \ldots, m \}$ and $t \in \{ 0, \ldots, W \}$, and hence the \textsc{Threshold Knapsack} problem, can be computed in time polynomial in $1/\varphi$ and the encoding length by the dynamic program in \Cref{alg:dp}.

    \begin{algorithm}[ht]
        \caption{A dynamic programming algorithm for \textsc{Threshold Knapsack}.}
        \label{alg:dp}
        \KwIn{a finite nonempty ground set $N$, $d : N \to \NN$, $w : N \to \NN$, $B \in \ZZ_+$, $\varphi \in \QQ \cap (0, 1)$ and $\Psi \in \QQ_{> 0}$ such that $d(e) \leq B$ for all $e \in N$.}
        \KwOut{either certify that there is no $S \in \mathcal F$ with $w(S) \geq \Psi$, or return some $S \in \mathcal F$ with $d(S) \leq d^\star(\Psi)$ and $w(S) \geq (1 - \varphi) \Psi$.}
        $m := |N|$, $M := (\varphi \Psi)/m$, $W := \lceil m(1 - \varphi)/\varphi \rceil$, $\hat w(e) := \lfloor \min\{ w(e), \Psi \}/M \rfloor$ for all $e \in N$. \\
        $\DP(0, 0) := 0$, $\DP(0, t) := \infty$ for all $t \in [W]$. \\
        \ForEach{$i := 1, \ldots, m$}{
            $\DP(i, t) := \DP(i - 1, t)$ for all $t \in \{ 0, \ldots, W \}$. \\
            \ForEach{$t := 0, \ldots, W$}{
                $t' := \min\{ W, t + \hat w(e_i) \}$. \\
                $\DP(i, t') := \min\{ \DP(i, t'), \DP(i - 1, t) + d(e_i) \}$.
            }
        }
        \uIf{$\DP(m, W) > B$}{
            \Return{there exists no $S \in \mathcal F$ with $w(S) \geq \Psi$}
        }
        \Else{
            Let $S \subseteq N$ be a subset attaining $\DP(m, W)$ by following the predecessors. \\
            \Return{$S$}
        }
    \end{algorithm}

    \Cref{alg:dp} performs $O(m(W + 1)) = O(|N|^2/\varphi)$ arithmetic operations, and all table entries and indices have polynomial encoding length. Hence, \Cref{alg:dp} runs in time polynomial in $1/\varphi$ and the encoding length of $(N, d, w, B, \varphi, \Psi)$.
    
    Suppose that some $S \in \mathcal F$ satisfies $w(S) \geq \Psi$. Then $d(S) \leq B$, and $\sum_{e \in S} \min\{ w(e), \Psi \} \geq \Psi$. Hence,
    \begin{align*}
        \hat w(S) &= \sum_{e \in S} \left\lfloor \frac{\min\{ w(e), \Psi \}}{M} \right\rfloor \geq \sum_{e \in S} \left( \frac{\min\{ w(e), \Psi \}}{M} - 1 \right) = \frac{\sum_{e \in S} \min\{ w(e), \Psi \}}{M} - |S| \\
        &\geq \frac{m\Psi}{\varphi \Psi} - |S| \geq \frac{m}{\varphi} - m = \frac{m(1 - \varphi)}{\varphi}.
    \end{align*}
    Since $\hat w(S) \in \ZZ_+$, we have $\hat w(S) \geq \lceil m(1 - \varphi)/\varphi \rceil = W$. Thus, $\DP(m, W) > B$ implies that there is no $S \in \mathcal F$ such that $w(S) \geq \Psi$. In particular, if $S^\star \in \mathcal F$ satisfies $d(S^\star) = d^\star(\Psi)$ and $w(S^\star) \geq \Psi$, then $\DP(m, W) \leq d(S^\star) = d^\star(\Psi)$. Hence, $\DP(m, W) \leq d^\star(\Psi)$ in all cases, since the inequality is automatic when $d^\star(\Psi) = \infty$.
    
    On the other hand, if $\DP(m, W) \leq B$ and if $S \subseteq N$ is returned by \Cref{alg:dp}, then $d(S) = \DP(m, W)$ and $\min\{ \hat w(S), W \} = W$, which implies that
    \begin{align*}
        w(S) &\geq \sum_{e \in S} \min\{ w(e), \Psi \} = M\sum_{e \in S} \frac{\min\{ w(e), \Psi \}}{M} \geq M \hat w(S) \geq MW \\
        &= \frac{\varphi \Psi}{m} \cdot \left\lceil \frac{m(1 - \varphi)}{\varphi} \right\rceil \geq \frac{\varphi \Psi}{m} \cdot \frac{m(1 - \varphi)}{\varphi} = (1 - \varphi) \Psi.
    \end{align*}
    Since $\DP(m, W) \leq B$ and since $\DP(m, W) \leq d^\star(\Psi)$, we have $d(S) \leq B$ and $d(S) \leq d^\star(\Psi)$. Since $w(S) \geq (1 - \varphi) \Psi > 0$, we have $S \neq \emptyset$ and hence $S \in \mathcal F$. This completes the proof.
\end{proof}

\Cref{lem:threshold}, together with the characterization of \citet{herzel2021one}, proves \Cref{lem:pareto}.

%% file: fptas.tex
\section{An FPTAS for Multiple-Choice Knapsack Intersection LP} \label{apx:fptas}

In this appendix, we prove \Cref{lem:fptas}. Let $G = (V, E)$ be a (multi)graph. Let $b : V \to \ZZ_+$. Let $d : E \to \NN$. Let $w : E \to \ZZ_+$. Without loss of generality, we assume that $G$ does not contain self-loops. Recall that, for each $u, v \in V$, $F_{u, v} \subseteq E$ is the subset of edges with endpoints $u$ and $v$, and that $\Pi_v := \{ F_{u, v} : u \in N(v) \}$ is a partition of $\delta(v)$ for each $v \in V$. Moreover, recall that, for each $v \in V$,
$$ \mathcal K_{\mathsf{mc}}(v) := \left\{ S \subseteq E : d(S \cap \delta(v)) \leq b(v), |S \cap \pi| \leq 1 \;\forall \pi \in \Pi_v \right\}. $$
Then \eqref{eq:mc-knapsack-lp} is equivalent to the following LP:
\begin{alignat}{4}
    \maximize \qquad && \sum_{e \in E} x_e w(e) \tag{Primal-Config-LP} \label{eq:config-lp} \\
    \subjectto \qquad && \sum_{S \in \mathcal K_{\mathsf{mc}}(v)} \lambda_{v, S} &= 1 && \qquad \forall v \in V, \notag \\
    && \sum_{\substack{S \in \mathcal K_{\mathsf{mc}}(v) \\ e \in S}} \lambda_{v, S} &= x_e && \qquad \forall v \in V, e \in \delta(v), \notag \\
    && x_e &\geq 0 && \qquad \forall e \in E, \notag \\
    && \lambda_{v, S} &\geq 0 && \qquad \forall v \in V, S \in \mathcal K_{\mathsf{mc}}(v). \notag
\end{alignat}
The dual LP of \eqref{eq:config-lp} is the following LP:
\begin{alignat}{4}
    \minimize \qquad && \sum_{v \in V} \mu_v \tag{Dual-Config-LP} \label{eq:dual-lp} \\
    \subjectto \qquad && \nu_{u, e} + \nu_{v, e} &\geq w(e) && \qquad \forall e = uv \in E, \notag \\
    && \mu_v &\geq \sum_{e \in S \cap \delta(v)} \nu_{v, e} && \qquad \forall v \in V, S \in \mathcal K_{\mathsf{mc}}(v), \notag \\
    && \mu_v &\in \RR && \qquad \forall v \in V, \notag \\
    && \nu_{v, e} &\in \RR && \qquad \forall v \in V, e \in \delta(v). \notag
\end{alignat}
We use $\LPOPT_{\mathsf{P}}$ and $\LPOPT_{\mathsf{D}}$ to denote the optimum of \eqref{eq:config-lp} and the optimum of \eqref{eq:dual-lp}, respectively. Let $(\bar x, \bar \lambda)$ be defined by $\bar x_e := 0$ for each $e \in E$, and $\bar \lambda_{v, S} := 1$ if $S = \emptyset$ and $\bar \lambda_{v, S} := 0$ otherwise. Let $(\bar \mu, \bar \nu)$ be defined by $\bar \nu_{u, e} := w(e)/2$ and $\bar \nu_{v, e} := w(e)/2$ for each $e = uv \in E$, and $\bar \mu_v := \sum_{e \in \delta(v)} w(e)/2$ for each $v \in V$. It is not hard to see that $(\bar \mu, \bar \nu)$ is a feasible solution of \eqref{eq:dual-lp}, and that $(\bar x, \bar \lambda)$ is a feasible solution of \eqref{eq:config-lp}. Hence, strong LP duality holds, i.e., $\LPOPT_{\mathsf{P}} = \LPOPT_{\mathsf{D}}$.

Given $\beta > 0$, we say that a \emph{$\beta$-approximate separation algorithm} for \eqref{eq:dual-lp} is an algorithm that, given a point $(\mu, \nu)$, either returns a violated constraint of \eqref{eq:dual-lp}, or guarantees that $(\mu/\beta, \nu)$ is a feasible solution of \eqref{eq:dual-lp}. We have the following lemma, which is analogous to \cite[Lemma 2.2]{FleischerLisa2011TAAf}.

\begin{lemma} \label{lem:approx}
    Suppose that there is a $\beta$-approximate separation algorithm $\mathcal A$ for \eqref{eq:dual-lp} for some $\beta \in (0, 1]$. Then for all $\kappa \in (0, \beta)$, one can compute, in time polynomial in $1/\kappa$, the encoding length and the running time of $\mathcal A$, a feasible solution $(x, \lambda)$ of \eqref{eq:config-lp} with polynomial-size support such that
    \begin{equation} \label{eq:approx}
        \sum_{e \in E} x_e w(e) \geq (\beta - \kappa) \cdot \LPOPT_{\mathsf{P}}.
    \end{equation}
\end{lemma}

Assuming \Cref{lem:approx}, we are ready to prove \Cref{lem:fptas}.

\begin{proof}[Proof of \Cref{lem:fptas}]
    Let $\varepsilon \in (0, 1)$. Let $\beta := 1 - \varepsilon/2$. First, we give a $\beta$-approximate separation algorithm for \eqref{eq:dual-lp}. Separating over constraints $\nu_{u, e} + \nu_{v, e} \geq w(e)$ for all $e = uv \in E$ can be done in polynomial time. If any of these constraints is violated, return it. Fix $v \in V$. Let
    $$ H_v := \max\left\{ \sum_{e \in S \cap \delta(v)} \nu_{v, e} : S \in \mathcal K_{\mathsf{mc}}(v) \right\}. $$
    Since $\mathcal K_{\mathsf{mc}}(v)$ is downward-closed, items $e \in \delta(v)$ with $\nu_{v, e} \leq 0$ may be discarded without changing the optimum. Applying an FPTAS for the \textsc{$0$-$1$ Multiple-Choice Knapsack} problem \citep[e.g.,][]{lawler1979fast,bansal2004improved} yields a set $S_v \in \mathcal K_{\mathsf{mc}}(v)$ such that $S_v \subseteq \delta(v)$ and that
    $$ \sum_{e \in S_v} \nu_{v, e} \geq \beta H_v. $$
    If $\sum_{e \in S_v} \nu_{v, e} > \mu_v$, return the violated constraint indexed by $(v, S_v)$. Otherwise,
    $$ \frac{\mu_v}{\beta} \geq \frac{1}{\beta} \sum_{e \in S_v} \nu_{v, e} \geq \frac{1}{\beta} \cdot \beta H_v = H_v. $$
    Hence, $(\mu/\beta, \nu)$ is feasible in \eqref{eq:dual-lp}. This gives a $\beta$-approximate separation algorithm for \eqref{eq:dual-lp}.
    
    By \Cref{lem:approx} with $\kappa := \varepsilon/2$, one can compute, in time polynomial in $1/\varepsilon$ and the input length, a feasible solution $(x, \lambda)$ of \eqref{eq:config-lp} with
    $$ \sum_{e \in E} x_e w(e) \geq (\beta - \kappa) \cdot \LPOPT_{\mathsf{P}} = \left(1 - \frac{\varepsilon}{2} - \frac{\varepsilon}{2}\right) \cdot \LPOPT_{\mathsf{P}} = (1 - \varepsilon) \cdot \LPOPT_{\mathsf{P}}. $$
    Since \eqref{eq:mc-knapsack-lp} is equivalent to \eqref{eq:config-lp}, this completes the proof.
\end{proof}

We now prove \Cref{lem:approx}. The proof is almost identical to that of Lemma 2.2 in \cite{FleischerLisa2011TAAf}, and this fact was also observed by \citet{carr2002randomized} and by \citet{644108.644154} for a class of packing-covering linear programs. We give its proof for completeness.

\begin{proof}[Proof of \Cref{lem:approx}]
    We run the ellipsoid algorithm on \eqref{eq:dual-lp} using $\mathcal A$. Specifically, given $T \geq 0$, we consider the following feasibility system $D(T)$:
    \begin{align*}
        \nu_{u,e}+\nu_{v,e}
            &\ge w(e)
            && \forall e=uv\in E, \\
        \mu_v
            &\ge \sum_{e\in S \cap \delta(v)}\nu_{v,e}
            && \forall v\in V,\ 
               S\in\mathcal K_{\mathsf{mc}}(v), \\
        \sum_{v\in V}\mu_v
            &\le T.
    \end{align*}
    Given $T \geq 0$, we use the ellipsoid algorithm to determine if $D(T)$ is feasible and use binary search to find the smallest feasible value $T$.

    Suppose the binary search algorithm terminates with a solution $(\mu^\star, \nu^\star)$ and $T^\star = \sum_{v \in V} \mu_v^\star$. Then $D(T^\star - \tau)$ is infeasible for some $\tau > 0$ depending on the precision of the binary search. Hence, the optimum of \eqref{eq:dual-lp} is at least $T^\star - \tau$. By the definition of a $\beta$-approximate separation algorithm, we have that $(\mu^\star/\beta, \nu^\star)$ is feasible. Hence, the optimum of \eqref{eq:dual-lp} is at least $T^\star - \tau$ and at most $T^\star/\beta$.

    In the execution of the ellipsoid algorithm for $T^\star - \tau$, we have checked a polynomial number of constraints. This set $\mathcal C \subseteq \{ (v, S) : v \in V, S \in \mathcal K_{\mathsf{mc}}(v) \}$ of constraints is sufficient to show that the optimum of \eqref{eq:dual-lp} is at least $T^\star - \tau$. Consider the following LP:
    \begin{alignat}{4}
        \minimize \qquad && \sum_{v \in V} \mu_v \tag{$\mathcal C$-Dual-Config-LP} \label{eq:dual-lp-restricted} \\
        \subjectto \qquad && \nu_{u, e} + \nu_{v, e} &\geq w(e) && \qquad \forall e = uv \in E, \notag \\
        && \mu_v &\geq \sum_{e \in S \cap \delta(v)} \nu_{v, e} && \qquad \forall (v, S) \in \mathcal C, \notag \\
        && \mu_v &\in \RR && \qquad \forall v \in V, \notag \\
        && \nu_{v, e} &\in \RR && \qquad \forall v \in V, e \in \delta(v). \notag
    \end{alignat}
    The dual LP of \eqref{eq:dual-lp-restricted} is the following LP obtained by restricting \eqref{eq:config-lp} to the $x$-variables and the $\lambda$-variables corresponding to constraints in $\mathcal C$, which has polynomially many variables and constraints:
    \begin{alignat}{4}
        \maximize \qquad && \sum_{e \in E} x_e w(e) \tag{$\mathcal C$-Primal-Config-LP} \label{eq:config-lp-restricted} \\
        \subjectto \qquad && \sum_{S \in \mathcal K_{\mathsf{mc}}(v)} \lambda_{v, S} &= 1 && \qquad \forall v \in V, \notag \\
        && \sum_{\substack{S \in \mathcal K_{\mathsf{mc}}(v) \\ e \in S}} \lambda_{v, S} &= x_e && \qquad \forall v \in V, e \in \delta(v), \notag \\
        && x_e &\geq 0 && \qquad \forall e \in E, \notag \\
        && \lambda_{v, S} &\geq 0 && \qquad \forall (v, S) \in \mathcal C. \notag
    \end{alignat}
    By strong LP duality, the optimum of \eqref{eq:config-lp-restricted} is at least $T^\star - \tau$. By choosing $\tau \leq \kappa T^\star/\beta$,
    $$ T^\star - \tau \geq T^\star - \frac{\kappa T^\star}{\beta} = (\beta - \kappa) \cdot \frac{T^\star}{\beta} \geq (\beta - \kappa) \cdot \LPOPT_{\mathsf{D}} = (\beta - \kappa) \cdot \LPOPT_{\mathsf{P}}. $$
    We solve \eqref{eq:config-lp-restricted} in polynomial time and extend its solution by setting $\lambda_{v, S} := 0$ for all $(v, S) \not \in \mathcal C$, thereby obtaining a feasible solution of \eqref{eq:config-lp} with the same objective value. This completes the proof.
\end{proof}